\documentclass{article}
\usepackage{amsmath,amsfonts,amsthm,amssymb}
\usepackage[margin=1in]{geometry}
\usepackage{url}
\usepackage{algorithm}
\usepackage{algpseudocode}
\usepackage{float}
\usepackage{enumerate}
\usepackage{epsfig}
\usepackage{subcaption}
\usepackage{color}
\usepackage[usenames,dvipsnames,svgnames,table]{xcolor}
\usepackage[colorlinks=true,
            linkcolor=red,
            urlcolor=blue,
            citecolor=gray]{hyperref}

\DeclareMathOperator*{\E}{\mathbb{E}}
\let\Pr\relax
\DeclareMathOperator*{\Pr}{\mathbb{P}}

\newcommand{\xend}{x_{\mathrm{end}}}
\newcommand{\yend}{y_{\mathrm{end}}}
\newcommand{\pend}{p_{end}}
\newcommand{\xbegin}{x_{\mathrm{begin}}}
\newcommand{\ybegin}{y_{\mathrm{begin}}}

\newcommand{\ol}{\overline}

\newcommand{\eqdef}{\mathbin{\stackrel{\rm def}{=}}}
\newcommand{\norm}[1]{\|#1\|}

\newcommand{\bv}[1]{\mathbf{#1}}

\makeatletter

\newtheorem*{rep@theorem}{\rep@title}
\newcommand{\newreptheorem}[2]{%
\newenvironment{rep#1}[1]{%
 \def\rep@title{#2 \ref{##1}}%
 \begin{rep@theorem}}%
 {\end{rep@theorem}}}
\makeatother
\newtheorem{theorem}{Theorem}
\newreptheorem{theorem}{Theorem}

\newtheorem{corollary}[theorem]{Corollary}
\newtheorem{lemma}[theorem]{Lemma}
\newtheorem*{lemma*}{Lemma}
\newreptheorem{lemma}{Lemma}

\newtheorem{claim}[theorem]{Claim}

\newif\ifpodc
\podctrue
\podcfalse

\begin{document}
\title{Ant-Inspired Density Estimation via Random Walks\footnote{A version of this work appears in the \emph{Proceedings of the National Academy of Sciences}, and is available at \url{http://www.pnas.org/content/114/40/10534}. An extended abstract initially appeared in the \emph{Proceedings of the 2016 ACM Symposium on Principles of Distributed Computing (PODC)}.}}
\author{Cameron Musco \\ MIT \\ \texttt{cnmusco@mit.edu} \\ \and
Hsin-Hao Su \\ UNC Charlotte\\ \texttt{hsinhaosu@uncc.edu} \\ \and
Nancy Lynch \\ MIT \\\texttt{lynch@csail.mit.edu}}

\maketitle

\begin{abstract}
Many ant species employ distributed population density estimation in applications ranging from quorum sensing \cite{pratt2005quorum}, to task allocation \cite{gordon1999interaction}, to appraisal of enemy colony strength \cite{adams1990boundary}. It has been shown that ants estimate density by tracking encounter rates -- the higher the population density, the more often the ants bump into each other \cite{pratt2005quorum,gordon1993function}.

We study distributed density estimation from a theoretical perspective. We prove that a group of anonymous agents randomly walking on a grid are able to estimate their density within a small multiplicative error in few steps by measuring their rates of encounter with other agents.  Despite dependencies inherent in the fact that nearby agents may collide repeatedly (and, worse, cannot recognize when this happens), our bound nearly matches what would be required to estimate density by independently sampling grid locations.

From a biological perspective, our work helps shed light on how ants and other social insects can obtain relatively accurate density estimates via encounter rates. From a technical perspective, our analysis provides new tools for understanding complex dependencies in the collision probabilities of multiple random walks. We bound the strength of these dependencies using \emph{local mixing properties} of the underlying graph. Our results extend beyond the grid to more general graphs and we discuss applications to size estimation for social networks and density estimation for robot swarms.
\end{abstract}

\thispagestyle{empty}
\clearpage
\setcounter{page}{1}

\section{Introduction}\label{sec:intro}

The ability to sense local population density is an important tool used by many ant species.
When a colony of \emph{Temnothorax} ants must relocate to a new nest, scouts search for potential nest sites, assess their quality, and recruit other scouts to high quality locations. A high enough density of scouts at a potential new nest (a \emph{quorum threshold}) triggers those ants to decide on the site and transport the rest of the colony there \cite{pratt2005quorum}. When neighboring colonies of \emph{Azteca} ants compete for territory, a high relative density of a colony's ants in a contested area will cause those ants to attack enemies in the area, while a low relative density will cause the colony to retreat \cite{adams1990boundary}. Varying densities of harvester ants successfully performing certain tasks such as foraging or brood care can trigger other ants to switch tasks, maintaining proper worker allocation in the colony \cite{gordon1999interaction,schafer2006forager}.

It has been shown that ants estimate density in a distributed manner, by measuring encounter rates \cite{pratt2005quorum,gordon1993function}. As ants randomly walk around an area, if they bump into a larger number of other ants, this indicates a higher population density. By tracking encounters with specific types of ants, for example, successful foragers or enemies, ants can estimate more specific densities. This strategy allows each ant to obtain an accurate density estimate and requires very little communication -- ants must simply detect when they collide and do not need to perform any higher level data aggregation.

\subsection{Density Estimation on the Grid}

We study distributed density estimation from a theoretical perspective. We model a colony of ants as a set of anonymous agents randomly placed on a two-dimensional grid. Computation proceeds in rounds, with each agent stepping in a random direction in each round. A \emph{collision} occurs when two agents reach the same position in the same round and encounter rate is measured as the number of collisions an agent is involved in during a sequence of rounds, divided by the number of rounds. Aside from collision detection, the agents have no other means of communication.

The intuition that encounter rate tracks density is clear. It is easy to show that, for a set of randomly walking agents, the \emph{expected} encounter rate measured by each agent is exactly the density $d$ -- the number of agents divided by the grid size (see Corollary \ref{cor:expectation}). However, it is unclear if the encounter rate actually gives a good density estimate -- that is, if the estimate is close to its expectation with high probability

Consider agents positioned not on the grid, but on a complete graph. In each round, each agent steps to a uniformly random position and, in expectation, the number of other agents it collides with in this step is $d$. Since each agent chooses its new location uniformly at random in each step, collisions are essentially \emph{independent} between rounds. The agents are effectively taking independent Bernoulli samples with success probability $d$, and by a standard Chernoff bound, within $O \left (\frac{\log(1/\delta)}{d\epsilon^2} \right)$ rounds each obtains a $(1\pm \epsilon)$ multiplicative approximation to $d$ with probability $1-\delta$.

On the grid graph, the picture is significantly more complex.
If two agents are initially located near each other, they are more likely to collide via random walking. After a first collision, due to their proximity, they are likely to collide repeatedly in future rounds. Since the agents are anonymous, they cannot recognize repeat collisions, and even if they could, it is unclear that it would help. On average, compared to the complete graph, agents collide with fewer individuals and collide multiple times with those individuals that they do encounter, making encounter rates a less reliable estimate of population density.

Mathematically speaking, on a graph with a \emph{fast mixing time} \cite{lovasz1993random}, like the complete graph, each agent's location is only weakly correlated with its previous locations.
This ensures that collisions are also weakly correlated between rounds and encounter rate serves as a very accurate estimate of density. The grid graph on the other hand is \emph{slow mixing} -- agent positions and hence collisions are highly correlated between rounds, lowering the accuracy of encounter-rate-based estimation.

\subsection{Our Contributions}

Surprisingly, despite the high correlation between collisions, we show that encounter-rate-based density estimation on the grid is nearly as accurate as on the complete graph. After just $O \left (\frac{\log(1/\delta)\cdot [\log\log(1/\delta)+\log(1/d\epsilon)]^2}{d\epsilon^2} \right)$ rounds, each agent's encounter rate is a $(1\pm \epsilon)$ approximation to $d$ with probability $1-\delta$ (Theorem \ref{naturalAlgoThm}). This matches performance on the complete graph up to a $[\log \log (1/\delta) + \log(1/d\epsilon)]^2$ factor.

Technically, to bound accuracy on the grid, we obtain moment bounds on the number of times that two randomly walking agents collide over a set of rounds (Lemma \ref{per_agent_moments}).
These bounds also apply to the number of equalizations (returns to origin) of a single walk.
While \emph{expected} random walk hitting times, return times, and collision rates are well studied for many graphs, including grid graphs \cite{lovasz1993random,elsasser2009tight,kanade2016coalescence}, higher moment bounds and high probability results are much less common.

Our moment bounds show that, while the grid graph is slow mixing, it has strong \emph{local mixing}. That is, random walks tend to spread quickly over a local area and not repeatedly cover the same nodes, making random-walk-based density estimation accurate. Significant work has focused on showing that random walk sampling is nearly as good as independent sampling for fast mixing expander graphs \cite{gillman1998chernoff,chung2012chernoff}. To the best of our knowledge, we are the first to extend this type of analysis to slowly mixing graphs, showing that strong local mixing is sufficient in many applications.

The key to the local mixing property of the grid is an upper bound on the probability that two random walks starting from the same position re-collide (or that a single random walk equalizes) after a certain number of steps (Lemma \ref{collideprobbound}). We show that re-collision probability bounds imply collision moment bounds on general graphs,
and apply this technique to extend our results to $d$-dimensional grids, regular expanders, and hypercubes. We discuss applications of our bounds to the task of estimating the size of a social network using random walks \cite{katzir2014estimating}, obtaining improvements over prior work for networks with relatively slow global mixing times but strong local mixing. We also discuss connections to density estimation by robot swarms and random-walk-based sensor network sampling \cite{avin2004efficient,lima2007random}. 

\subsection{Road Map}
\begin{description}
\item In Section \ref{sec:model} we overview our theoretical model for distributed density estimation on the grid.
\vspace{-.05em}
\item In Section \ref{sec:randomWalks} we give our main technical results on random-walk-based density estimation.
\vspace{-.05em}
\item In Section \ref{generalizations} we show how to extend our bounds to a number of graphs other than the grid. 
\vspace{-.05em}
\item In Section \ref{sec:applications} we discuss applications of our results to social network size estimation and robot swarm algorithms.
\vspace{-.05em}
\item In Section \ref{sec:futureWork} we conclude and discuss interesting open questions and directions for future work.
\end{description}

\section{Theoretical Model for Density Estimation}\label{sec:model}

We consider a set of agents populating a two-dimensional torus with $A$ nodes (dimensions $\sqrt{A} \times \sqrt{A}$). At each time step, each agent has an associated ordered pair $position$, which gives its coordinates on the torus.
We assume that $A$ is large -- larger than the area agents traverse over the runtimes of our algorithms. 
We believe the torus model successfully captures the dynamics of density estimation on a surface, while avoiding complicating factors of boundary behavior on a finite grid.

Initially each agent is placed independently at a uniform random node in the torus. 
Computation proceeds in discrete, synchronous rounds.
In each round, an agent may either remain in its current location or step to any of its four neighboring grid squares. Formally, it updates the ordered pair $position$ by adding a step chosen from $\{(0,1),(0,-1),(1,0),(-1,0), (0,0)\}$.

A \emph{randomly walking agent} chooses its step uniformly at random from $\{(0,1),(0,-1),(1,0),(-1,0)\}$ in each round. Of course, in reality ants do not move via pure random walk. However, there is evidence that nevertheless, encounter rates are well predicted by a random walk model \cite{boczkowski2017limits}. At the same time, there is evidence that in some cases, encounter rates are actually lower than predicted by such a model \cite{gordon1993function,nicolis2005effect}. Overall, we feel that our model sufficiently captures the highly random movement of ants while remaining tractable to analysis and applicable to ant-inspired random-walk-based algorithms (Section \ref{sec:applications}). Extending our work to more realistic models of ant movement would be an interesting next direction.

Aside from the ability to move in each round, agents can sense the number of agents other than themselves at their position at the \emph{end of each round}, formally through the function $count(position)$. We say that two agents \emph{collide in round $r$} if they have the same position at the end of the round.
Outside of collision counting, agents have no means of communication. 
They are anonymous (cannot uniquely identify each other) and execute identical density estimation routines. 
A basic illustration of our model is depicted in Figure \ref{fig:grid}.

\begin{figure}
\centering
\includegraphics[width=.4\linewidth]{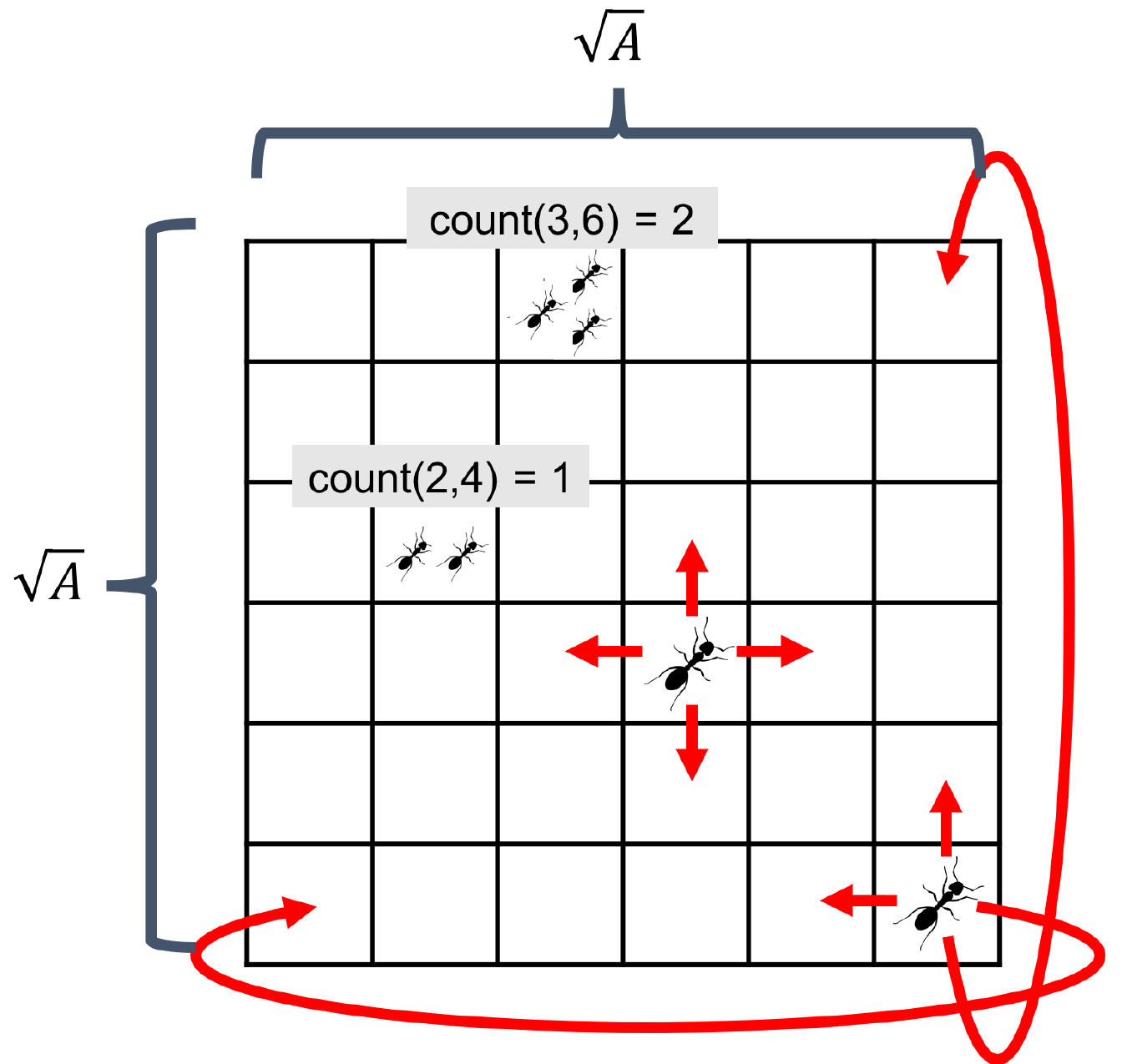}
\caption{A basic illustration of our computational model. Each agent (ant) may move to an adjacent position on the two-dimensional torus in each round (illustrated by the red arrows). A collision occurs when two or more agents are located at the same position. The agents detect collisions through the $count(position)$ function which returns the \emph{number of other agents} at their current position. In this illustration, $position$ is given as the $(x,y)$ position with the bottom left corner corresponding to $(1,1)$. However, the precise convention used is unimportant.}
\label{fig:grid}
\end{figure} 

\subsection{The Density Estimation Problem}\label{sec:problemD}
Let $(n+1)$ be the number of agents and define population density as $d \eqdef n/A$. Each agent's goal is to estimate $d$ to $(1 \pm \epsilon)$ accuracy with probability at least $1-\delta$ for $\epsilon, \delta \in (0,1)$ -- that is, to return an estimate $\tilde d$ with $\Pr \left [ \tilde d \in \left [(1-\epsilon)d, (1+\epsilon) d \right ] \right ] \ge 1-\delta$. As a technicality, with $n+1$ agents we define $d = n/A$ instead of $d = (n+1)/A$. In the natural case, when $n$ is large, the distinction is unimportant. Since our analysis always takes the perspective of one agent, this convention ensures that there are $n$ other agents with which this agent may interact, and thus all summations over expected collision counts and other quantities are over $n$ variables rather than $n-1$. Additionally, in the case there is a single agent on the grid, this convention allows the agent to return density estimate $0$. If the density were instead defined to be $1/A$ in this case, estimation would be impossible in our model, since the agent has no way of sensing its location and thus no way of estimating the size of the torus, $A$.

\subsubsection{Local vs. Global Density}

The problem described above requires estimating the \emph{global population density}. We assume that agents are initially distributed uniformly at random on the torus, which is critical for fast global density estimation. When agents are uniformly distributed, the local density in a small radius around their starting position reflects the global density with good probability. Thus, they are  able to obtain a good estimate of this density using local measurements, and without traversing a large fraction of the torus.
 
Of course, in nature, ants are not typically uniformly distributed in the nest or surrounding areas. Additionally, they are often interested in estimating \emph{local population densities} -- e.g., around a nest entrance when estimating the number of successful foragers for task allocation \cite{gordon1999interaction}. We view our work as a first step towards a theoretical understanding of density estimation and focus on the global density for simplicity. Removing our assumption of uniformly distributed agents, formally defining the problem of local density estimation, and understanding how ants can solve this problem are important directions for future work.

\section{Random-Walk-Based Density Estimation}\label{sec:randomWalks}

As discussed, the challenge in analyzing random-walk-based density estimation on the torus arises from correlations between collisions of nearby agents. If we do not restrict agents to random walking, and instead allow each agent to take an arbitrary step in each round, they can avoid collision correlations by splitting into `stationary' and `mobile' groups and counting collisions only between members of different groups. This allows them to essentially independently sample collisions with other agents to estimate density. This method is simple to analyze (see Appendix \ref{sec:independent}), but it is not `natural' in a biological sense or useful for the applications of Section \ref{sec:applications}. Further, independent sampling is unnecessary! Algorithm \ref{random_walk_sampling} describes a simple random-walk-based approach that gives a nearly matching bound.

\begin{algorithm}[H]
\caption{Random-Walk-Based Density Estimation}
Each agent independently executes:
\begin{algorithmic}
\State $c := 0$
\For{$r = 1,..., t$}
	\State{$step := rand \{(0,1),(0,-1),(1,0),(-1,0) \}$}
	\State{$position := position + step$}
	\State{$c := c+ count(position) $}
	\Comment{\textcolor{blue}{Update collision count.}}
\EndFor \\
\Return{$\tilde d = \frac{c}{t}$}
\end{algorithmic}
\label{random_walk_sampling}
\end{algorithm}

\subsection{Random-Walk-Based Density Estimation Analysis}

Our main theoretical result follows; its proof appears at the end of Section \ref{sec:randomWalks}, after a number of preliminary lemmas. Throughout our analysis, we take the viewpoint of a single agent executing Algorithm \ref{random_walk_sampling}, which we sometimes call \emph{agent $a$}.
\begin{theorem}[Random Walk Sampling Accuracy Bound]\label{naturalAlgoThm} After running for $t$ rounds, assuming $t \le A$, an agent executing Algorithm \ref{random_walk_sampling} returns $\tilde d$ such that, for any $\delta > 0$, with probability $\ge 1-\delta$, $$\tilde d \in [(1-\epsilon ) d,  (1+\epsilon) d ]\text{ for }\epsilon \le c_1 \cdot  \sqrt{\frac{\log(1/\delta)}{td}} \cdot \log(2t),$$ where $c_1$ is some fixed constant. This implies that, for any $\epsilon, \delta \in (0,1)$ if 
$$A \ge t \ge  \frac{c_2\log(1/\delta) \cdot \left  [\log\log(1/\delta) +\log(1/d\epsilon) \right ]^2}{d\epsilon^2},$$ where $c_2$ is some fixed constant, $\tilde d \in [(1-\epsilon ) d,  (1+\epsilon) d ]$ with probability $\ge 1-\delta$.
\end{theorem}
Theorem \ref{naturalAlgoThm} focuses on the density estimate of a single agent executing Algorithm \ref{random_walk_sampling}. However, we note that if we set $\delta = \frac{\delta'}{n}$, then by a union bound, all $n$ agents will have $\tilde d \in [(1-\epsilon)d,(1+\epsilon)d]$ with probability $\delta'$. The required running time $t$ will depend just logarithmically on $\delta'$ and $n$.

\subsection{Decomposition of Collision Count into Independent Random Variables}\label{sec:decomp}
We decompose the collision count $c$ maintained by an agent executing Algorithm \ref{random_walk_sampling} as the sum of collisions with different agents over different rounds. Specifically, assign arbitrary ids $1,2,...,n$ to the $n$ other agents  and let $c_j(r)$ equal $1$ if the agent collides with agent $j$ in round $r$, and $0$ otherwise. Let $c_j = \sum_{r=1}^t c_j(r)$ be the total number of collisions with agent $j$.
We have $c = \sum_{j=1}^n c_j$. Note that $c_1,...,c_n$ are identically distributed random variables.

The main challenge in proving the accuracy of Algorithm  \ref{random_walk_sampling} is in handling the strong correlations between collisions in successive rounds -- i.e., between the random variables $c_j(1),...,c_j(t)$ for each $j$. Across agents, the collision counts  $c_1,...,c_n$ may also be correlated. However, conditioned on the random walk taken by agent $a$ (the agent whose viewpoint we take), $c_1,...,c_n$ are independent, since they depend only  on the independent random walks of different agents.
Thus, the results in this section will typically bound collision probabilities and expectations conditioned on agent $a$'s path, which we denote by $\mathcal{W}$. $\mathcal{W}$ is a random variable, consisting of a sequence of $t$ positions. In Section \ref{sec:removeCond} we will remove this conditioning, showing that Algorithm \ref{random_walk_sampling} yields an accurate density estimate, regardless of agent $a$'s path, and thereby proving Theorem \ref{naturalAlgoThm}.

\subsection{Correctness of Encounter Rate in Expectation}

\begin{lemma}[Unbiased Estimator]\label{estimate_expectation}
Let $\mathcal{W}$ be the $t$-step random walk that an agent executing Algorithm \ref{random_walk_sampling} takes. The output $\tilde d$ of that agent  satisfies: $\E[\tilde d | \mathcal{W}]= d$.
\end{lemma}
\begin{proof}
By linearity of expectation,
$\E [c | \mathcal{W}] = \sum_{j = 1}^n \sum_{r = 1}^t \E [c_j(r)  | \mathcal{W}].$ Conditioned on $\mathcal{W}$, the position of the agent is fixed in round $r$. Since each other agent is initially at a uniform random location and after any number of steps, is still at uniform random location, for all $j,r$, $\E [c_j(r) \mid \mathcal{W}] = 1/A$. Thus, $\E [c \mid \mathcal{W}] = nt/A = dt$ and $\E [\tilde d \mid \mathcal{W}] = \E [c \mid \mathcal{W}]]/t =  d$.
\end{proof}
By the law of iterated expectation, $\E [\tilde d]= \E[ \E [\tilde d | \mathcal{W}]]$  and so Lemma \ref{estimate_expectation} gives:
\begin{corollary}\label{cor:expectation} $\E [\tilde d] =d $.
\end{corollary}

We note that the torus is bipartite, and hence two agents initially located an odd number of steps away from each other will never meet via random walking. However, this fact does not change the expectation of $\tilde d$ computed above and in fact does not affect any of our following proofs. 

We note that the torus is bipartite, and hence two agents initially located an odd number of steps away from each other will never meet via random walking. However, this fact does not change the expectation of $\tilde d$ computed above and in fact does not affect any of our following proofs. 

With Lemma \ref{estimate_expectation} and Corollary \ref{cor:expectation} in place, it remains to show that the encounter rate is close to its expectation with high probability and so provides a good estimate of density. In order to do this,  we must bound the strength of correlations between collisions of nearby agents in successive rounds, which can decrease the accuracy of the encounter-rate-based estimate.

\subsection{A Re-collision Probability Bound}

The key to bounding collision correlations is bounding the probability of a re-collision between two randomly walking agents in round $r+m$, assuming a collision in round $r$, which we do in Lemma \ref{collideprobbound} below.\footnote{In fact, we prove a stronger result, giving a bound on the re-collision probability conditioned on the random walk taken by one of the agents in rounds $r+1,...,r+m$. As discussed in Section \ref{sec:decomp}, it will later be necessary to condition on this walk to ensure that the number of collisions between the agent  and each other agent (i.e., $c_1,...,c_n$) are independent.}
Each $c_j$ is the sum of highly correlated random variables $c_j(1),...,c_j(t)$. Due to the slow mixing of the grid, if two agents collide at round $r$, they are much more likely to collide in successive rounds. However, by bounding this re-collision probability, we are able to give strong moment bounds for the distribution of each $c_j$. We bound not only its variance, but all higher moments. This allows us to show that the average $\tilde d = \frac{1}{t} \sum_{j=1}^n c_j$ falls close to its expectation $d$ with high probability, giving Theorem \ref{naturalAlgoThm}.

Our re-collision probability bound is stated below:

\begin{lemma}[Re-collision Probability Bound]\label{collideprobbound} Consider two agents $a_1$ and $a_2$ randomly walking on a two-dimensional torus of dimensions $\sqrt{A}\times \sqrt{A}$. Assume that $a_1$ and $a_2$ collide in round $r$. For any $m \ge 0$, let $\mathcal{W}$ be the $m$-step random walk performed by $a_2$ in rounds $r+1,...,r+m$. Let $\mathcal{C}$ be the event that $a_1$ and $a_2$ collide again in round $r+m$. We have:
\begin{align*}
\Pr[\mathcal{C} | \mathcal{W}] = O\left (\frac{1}{m+1} +\frac{1}{A}\right ). 
\end{align*}
\end{lemma}

\medskip
\noindent\textbf{Lemma \ref{collideprobbound} Proof Outline.}
\medskip

\noindent Our proof of Lemma \ref{collideprobbound} in broken down into the following steps.   See Figure \ref{fig:proof}  for a schematic of the proof.
\begin{figure}[h]
 \centering
\includegraphics[width=.4\linewidth]{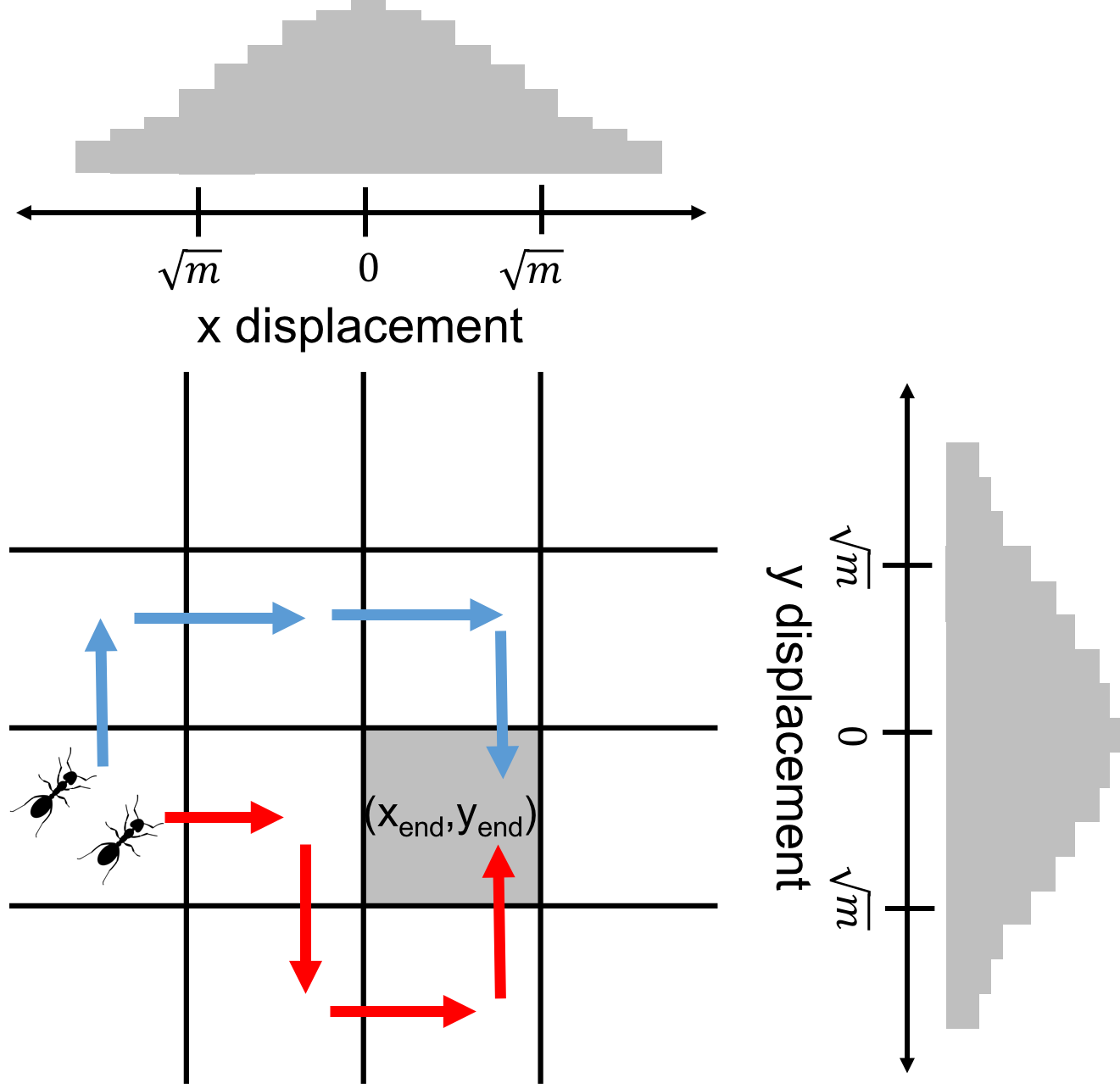}
\caption{
A schematic of the proof of Lemma \ref{collideprobbound}.}
 \label{fig:proof}
\end{figure} 
 \begin{enumerate}
 \item  In Lemma \ref{lem:hitting} we bound the probability that a single $m$-step random walk starting from some position ends at any other $x$ or $y$  position, conditioned on the number of steps that the walk takes in the $x$ and $y$ directions. The proof of this lemma breaks down into two cases:
  \begin{itemize}
  \item In Claim \ref{clm:c1} we show that if the walk takes $m_x$ steps in the $x$ direction and $m_y$ steps in the $y$ direction and does not fully `wrap around' the torus, the probability that it ends at any $x$ position can be bounded by $O \left (\frac{1}{\sqrt{m_x}} \right)$. The probability that it  ends at any $y$ position can similarly  be bounded by  $O \left (\frac{1}{\sqrt{m_y}} \right)$.
 \item In Claim \ref{clm:c2} we handle the case when the walk does wrap fully around the torus one or more times, showing that this possibility adds at most an additional $O \left  (\frac{1}{\sqrt{A}} \right)$ factor to the probability of ending at any $x$ or $y$ position on a $\sqrt{A} \times \sqrt{A}$ torus. This gives an overall bound on the probability of ending at any $x$ position of $O \left (\frac{1}{\sqrt{m_x}} + \frac{1}{\sqrt{A}} \right)$, and an analogous bound on the probability of ending at any $y$  position of $O \left (\frac{1}{\sqrt{m_y}} + \frac{1}{\sqrt{A}} \right)$.
 \end{itemize}
 \item In Corollary  \ref{cor:hitting} we show that, since movement in the $x$ and $y$ directions are independent, the probability of a single walk ending at  any position after taking $m_x$ steps in the $x$ direction and $m_y$ steps in the $y$ direction is:
 $$O \left (\left[\frac{1}{\sqrt{m_x}} + \frac{1}{\sqrt{A}}\right] \cdot \left[\frac{1}{\sqrt{m_y}} + \frac{1}{\sqrt{A}}\right]  \right) = O \left (\frac{1}{\sqrt{m_x \cdot m_y}} + \frac{1}{A} \right).$$
 \item In Lemma \ref{lem:hitting2} we show that, with high probability, an $m$-step walk takes $\Theta(m)$ steps in both the $x$ and $y$ directions. Combined with Corollary \ref{cor:hitting}, this yields an unconditional bound of $O \left (\frac{1}{m}+\frac{1}{A}\right)$ on the probability that a  single random walk starting from some position ends at any other particular position after $m$ steps.
 \item Finally, we bound the probability that two walks re-collide after $m$ steps, conditioned on the path of one of these walks,  giving Lemma \ref{collideprobbound}. Fixing this path fixes a position $(\xend,\yend)$ that the agent is located in at  round $r+m$. For a re-collision to occur, the other agent must also be located at this position at round $r+m$. We bound the probability  of this event directly with the single walk bound of Lemma \ref{lem:hitting2}.
\end{enumerate}

We begin with our bound on the probability  of a single walk ending at any  $x$ and $y$  position, conditioned on the number number of steps that it takes in each direction.
\begin{lemma}\label{lem:hitting}
Consider an agent $a_1$ randomly  walking on a two-dimensional torus of dimensions $\sqrt{A}\times \sqrt{A}$ which is at position $(\xbegin,\ybegin)$ in round $r$. For any $m \ge 0$ and any position $(\xend,\yend)$, let $\mathcal{C}_{x}$ be the event  that $a_1$ has $x$ position $\xend$ at round $r+m$ and let $\mathcal{C}_y$ be the event that $a_1$ has $y$ position $\yend$ at round $r+m$. Let $M_x,M_y$ be random variables giving the number of steps that $a_1$ takes in the $x$ and $y$ directions respectively in rounds $r+1,...r+m$. For any $m_x,m_y \in \{0,...,m\}$:
\begin{align*}
\Pr[\mathcal{C}_{x} | M_x = m_x] = O\left ( \frac{1}{\sqrt{m_x+1}} + \frac{1}{\sqrt{A}}\right)
\end{align*}
and 
\begin{align*}
\Pr[\mathcal{C}_{y} | M_y = m_y] = O\left ( \frac{1}{\sqrt{m_y+1}} + \frac{1}{\sqrt{A}}\right).
\end{align*}
\end{lemma}
\begin{proof}
We focus on bounding $\Pr \left [ \mathcal{C}_x | M_x = m_x \right ]$. The bound on $\Pr \left [ \mathcal{C}_y | M_y = m_y \right ]$ follows from an identical proof.
We split our analysis into two cases. Let $\delta_x = \xend-\xbegin$  be the change  in $x$ position required for $\mathcal{C}_x$ to occur.
Let $\mathcal{C}_x^1$ be the event that 
$a_1$  has total $x$ displacement $\delta_x$ from round $r$ to round $r+m$ (and so is at  $x$ position $\xend$ in round $r+m$).
Let $\mathcal{C}_x^2$ be the event that the agent is at $x$ position $\xend$ in round $r+m$ but \emph{does not} have displacement $\delta_x$. This requires that the agent `wraps around' the torus, ending at $\xend$ despite moving further than $\delta_x$. We can write: 
\begin{align}\label{csplitnero}
\Pr[\mathcal{C}_x| M_x = m_x] = \Pr[\mathcal{C}_x^1| M_x = m_x]  + \Pr[\mathcal{C}_x^2| M_x = m_x].
\end{align}
We bound the probabilities of $\mathcal{C}_x^1$ and $\mathcal{C}_x^2$ separately.
\begin{claim}[Collision Probability Without Wraparound]\label{clm:c1}
\begin{align*}
\Pr[\mathcal{C}_x^1| M_x = m_x] = O \left (\frac{1}{\sqrt{m_x+1}}\right ).
\end{align*}
\end{claim}
\begin{proof}
We can write the $x$ displacement of $a_1$ as $\sum_{j=1}^{m_x} s_j$ where $s_j$ the direction of the agent's $j^{th}$ step in the $x$ direction. Each $s_j$ is an independent random variable equal to 1 with probability $1/2$ and $-1$ with probability $1/2$. With this notation we can compute:
%
\begin{align}
\Pr[\mathcal{C}_x^1| M_x = m_x] &= 
 \Pr \left [\left (\sum_{j=1}^{m_x} s_j\right) = \delta_x | M_x = m_x \right] =  \binom{m_x}{\frac{m_x+\delta_x}{2}}\left(\frac{1}{2}\right)^{m_x}
 \label{preStir}
\end{align}
where we use the convention that the binomial coefficient equals $0$ if $\frac{m_x+\delta_x}{2}$ is not a positive integer. For any $\delta_x$
we have $\binom{m_x}{\frac{m_x+\delta_x}{2}} \le \binom{m_x}{\lfloor m_x/2\rfloor} = \frac{m_x!}{\lfloor \frac{m_x}{2}\rfloor!\cdot \lceil \frac{m_x}{2}\rceil!} $ and by Stirling's approximation, for any $n > 0$, $n! = \sqrt{2\pi n} \left (\frac{n}{e} \right )^n \left (1+O\left (\frac{1}{n} \right )\right)$, which gives:
\begin{align*}
\Pr[\mathcal{C}_x^1| M_x = m_x] = O \left (\frac{1}{\sqrt{m_x+1}}\right ),
\end{align*}
and so the claim. Note that
we use $m_x+1$ instead of $m_x$ in the denominator so that the bound is meaningful in the case when $m_x = 0$.
\end{proof}
We next show:

\begin{claim}[Collision Probability with Wraparound]\label{clm:c2}
\begin{align*}
\Pr \left [ \mathcal{C}_x^2 | M_x = m_x \right ] = O \left (\frac{1}{\sqrt{m_x+1}} + \frac{1}{\sqrt{A}} \right )
\end{align*}
\end{claim}
\begin{proof}
In order for $\mathcal{C}_x^2$ to occur, $a_1$ must have $x$ position $\xend$ after round $r+m$ but not have total displacement  $\delta_x$. In particular, $a_1$'s displacement must differ from $\delta_x$ by a nonzero integer multiple of $\sqrt{A}$ -- the side length of the torus. We can thus write, letting $\mathbb{Z}\setminus 0$ denote the set of nonzero integers:
\begin{align}
\Pr \left [ \mathcal{C}_x^2 | M_x = m_x \right ] &= \sum_{c\in \mathbb{Z}\setminus 0} \Pr \left [\left (\sum_{j=1}^{m_x} s_j \right ) = \delta_x + c\sqrt{A} | M_x = m_x \right ]\nonumber \\
&= \left  (\frac{1}{2}\right)^{m_x} \cdot \sum_{c\in \mathbb{Z}\setminus\{0,\pm1\}} {m_x \choose \frac{m_x + \delta_x + c\sqrt{A}}{2}}+ O  \left ( \frac{1}{\sqrt{m_x+1}} \right )\label{tBound}.
\end{align}
To obtain \eqref{tBound} we bound the $c= \pm 1$ terms $\Pr \left [\left (\sum_{j=1}^{m_x} s_j \right) = \delta_x \pm \sqrt{A} | M_x = m_x \right ]$ by $O \left (\frac{1}{\sqrt{m_x+1}} \right )$ using our bound on $\mathcal{C}_x^1$ given in Claim \ref{clm:c1}, which can easily be seen to hold for any $x$ displacement, and in particular, for $\delta_x \pm \sqrt{A}$.

Roughly, we will upper  bound the first term of \eqref{tBound} by the probability that the agent ends at \emph{any} $x$ position in round $r+m$. Since there are $\sqrt{A}$ such positions, this probability is thus bounded by $O\left  (\frac{1}{\sqrt{A}}\right)$.
Formally, consider any  $i \in [1,...,\sqrt{A}-1]$ and let $\mathcal{D}^i_x$ be the event that the walk is $i$ steps clockbywise from $\xend$ after taking $M_x$ steps. We can write:
\begin{align}\label{nozeroDisp}
\Pr [\mathcal{D}^i_x | M_x = m_x] &= \left (\frac{1}{2}\right )^{m_x} \cdot \sum_{c \in \mathbb{Z}} {m_x \choose \frac{m_x + \delta_x + i + c\sqrt{A}}{2}}.
\end{align}
Now, since $|\delta_x| < \sqrt{A}$ and $i < \sqrt{A}$, for $c \ge 2$, 
$\delta_x + i + (c-1)\sqrt{A}$ is closer to $0$ than $\delta_x +  c\sqrt{A}$. So, as long as $\frac{m_x + \delta_x + i + c\sqrt{A}}{2}$ is an integer,
$
{m_x \choose \frac{m_x + \delta_x + i + (c-1)\sqrt{A}}{2}} \ge {m_x \choose \frac{m_x + \delta_x + c\sqrt{A}}{2}}.
$
Similarly, for $c \le -2$, $\delta_x + i + c\sqrt{A}$ is closer to $0$ than $\delta_x + c\sqrt{A}$. So as long as $\frac{m_x + \delta_x + i + c\sqrt{A}}{2}$ is an integer, 
$
{m_x \choose \frac{m_x + \delta_x + i + c\sqrt{A}}{2}} \ge {m_x \choose \frac{m_x + \delta_x + c\sqrt{A}}{2}}.
$

Let  $\mathcal{E}_{i,c}$ equal $1$ if $\frac{m_x + \delta_x + i + c\sqrt{A}}{2}$ is an integer and $0$ otherwise. By the above bounds we have:
\begin{align*}
\sum_{c \in \mathbb{Z} \setminus \{0,\pm1\}} \mathcal{E}_{i,c} \cdot  {m_x \choose \frac{m_x + \delta_x + c\sqrt{A}}{2}} &\le \sum_{c \in \mathbb{Z} \setminus \{0,- 1\}} \mathcal{E}_{i,c} \cdot  {m_x \choose \frac{m_x + \delta_x + i + c\sqrt{A}}{2}}  \\&\le \sum_{c \in \mathbb{Z}} \mathcal{E}_{i,c} \cdot  {m_x \choose \frac{m_x + \delta_x + i + c\sqrt{A}}{2}}.
\end{align*}

Combining with \eqref{nozeroDisp} and using the fact that the $\mathcal{D}^i_x$ events are  disjoint, so the sum of their probabilities is at most $1$, we have:
\begin{align*}
\sum_{i=1}^{\sqrt{A}-1} \Pr \left [ \mathcal{D}_x^i | M_x = m_x \right ]  &\le 1\\
\sum_{i=1}^{\sqrt{A}-1} \left [\left  (\frac{1}{2} \right )^{m_x}\cdot \sum_{c \in \mathbb{Z}} \mathcal{E}_{i,c} \cdot  {m_x \choose \frac{m_x + \delta_x + i + c\sqrt{A}}{2}} \right ]  &\le 1
\\
\sum_{i=1}^{\sqrt{A}-1} \left [
\left  (\frac{1}{2} \right )^{m_x} \sum_{c \in \mathbb{Z} \setminus \{0,\pm1\}} \mathcal{E}_{i,c} \cdot {m_x \choose \frac{m_x + \delta_x + c\sqrt{A}}{2}} \right ]  &\le 1
\\
\left  (\frac{1}{2} \right )^{m_x} \sum_{c \in \mathbb{Z} \setminus \{0,\pm1\}} \left [ \left (\sum_{i=1}^{\sqrt{A}-1} \mathcal{E}_{i,c}\right) \cdot {m_x \choose \frac{m_x + \delta_x + c\sqrt{A}}{2}} \right ] &\le 1.
\end{align*}
Now, for all $c$, $\sum_{i=1}^{\sqrt{A}-1}  \mathcal{E}_{i,c} = \Theta \left ( \sqrt{A} \right )$ since $\frac{m_x + \delta_x + i  + c\sqrt{A}}{2}$ is integral for half the possible $i \in [1,...,\sqrt{A}-1]$. Rearranging, we thus have $\left  (\frac{1}{2} \right )^{m_x} \sum_{c \in \mathbb{Z} \setminus \{0,\pm1\}} {m_x \choose \frac{m_x + \delta_x + c\sqrt{A}}{2}} = O\left  (\frac{1}{\sqrt{A}}\right).$ Plugging this back into \eqref{tBound}:
\begin{align*}
\Pr \left [ \mathcal{C}_x^2 | M_x = m_x \right ] = O \left (\frac{1}{\sqrt{m_x+1}} + \frac{1}{\sqrt{A}} \right ),
\end{align*}
which gives the claim.
\end{proof}

Plugging the bounds of Claims \ref{clm:c1} and \ref{clm:c2} into \eqref{csplitnero} we have:
\begin{align*}
\Pr \left [ \mathcal{C}_x | M_x = m_x \right ]  &= \Pr \left [ \mathcal{C}_x^1 | M_x = m_x \right ]  + \Pr \left [ \mathcal{C}_x^2 | M_x = m_x \right ] \\
& O\left ( \frac{1}{\sqrt{m_x+1}} + \frac{1}{\sqrt{A}} \right ),
\end{align*}
which gives the lemma.
\end{proof}
Since an agent's movements in the $x$ and $y$ directions are independent,
Lemma \ref{lem:hitting}  yields:
\begin{corollary}\label{cor:hitting}
Consider an agent $a_1$ randomly  walking on a two-dimensional torus of dimensions $\sqrt{A}\times \sqrt{A}$ which is at position $(\xbegin,\ybegin)$ in round $r$. For any $m \ge 0$ and any position $(\xend,\yend)$, let $\mathcal{C}$ be the event  that $a_1$ has $x$ position $(\xend,\yend)$ at round $r+m$. Let $M_x,M_y$ be random variables giving the number of steps that $a_1$ takes in the $x$ and $y$ directions respectively in rounds $r+1,...r+m$. For any $m_x,m_y \in \{0,...,m\}$:
\begin{align*}
\Pr[\mathcal{C} | M_x = m_x, M_y = m_y] = O\left ( \frac{1}{\sqrt{(m_x+1)(m_y+1)}} + \frac{1}{A}\right).
\end{align*}
\end{corollary}
\begin{proof}
Let $\mathcal{C}_x$ and $\mathcal{C}_y$ be as defined in Lemma \ref{lem:hitting}.
All steps are chosen independently, so conditioned on $M_x = m_x$ and $M_y = m_y$, $\mathcal{C}_x$ and $\mathcal{C}_y$ are independent. We can thus compute: 
\begin{align*}
\Pr \left [\mathcal{C} | M_x = m_x, M_y = m_y\right ] &= \Pr \left [\mathcal{C}_x\text{ and }\mathcal{C}_y | M_x = m_x, M_y = m_y\right ] \\
&= \Pr \left [ \mathcal{C}_x | M_x = m_x \right ] \cdot \Pr \left [ \mathcal{C}_y | M_y = m_y \right ].
\end{align*}
Applying the bounds from Lemma \ref{lem:hitting} we have:
\begin{align*}
\Pr \left [\mathcal{C} | M_x = m_x, M_y = m_y\right ] &= O \left (\left [\frac{1}{\sqrt{m_x+1}}+\frac{1}{\sqrt{A}}\right ] \cdot \left [\frac{1}{\sqrt{m_y+1}}+\frac{1}{\sqrt{A}}\right ]  \right )\\
&= O\left ( \frac{1}{\sqrt{(m_x+1)(m_y+1)}} + \frac{1}{A}\right),
\end{align*}
which gives the corollary.
\end{proof}
Using Corollary \ref{cor:hitting} we can give an unconditional bound on the collision probability by showing that, with high probability, $m_x = \Theta(m)$ and $m_y = \Theta(m)$.
\begin{lemma}[Single Random Walk Collision Probability]\label{lem:hitting2}
Consider an agent $a_1$ randomly  walking on a two-dimensional torus of dimensions $\sqrt{A}\times \sqrt{A}$ which is at position $(\xbegin,\ybegin)$ in round $r$. For any $m \ge 0$ and any position $(\xend,\yend)$, let $\mathcal{C}$ be the event  that $a_1$ has $x$ position $(\xend,\yend)$ at round $r+m$.
\begin{align*}
\Pr[\mathcal{C}] = O\left ( \frac{1}{m+1} + \frac{1}{A}\right).
\end{align*}
\end{lemma}
\begin{proof}
As in Corollary  \ref{cor:hitting}, let $M_x,M_y$ be random variables giving the number of steps that $a_1$ takes in the $x$ and $y$ directions respectively in rounds $r+1,...r+m$. Since direction is chosen independently and uniformly at random for each step, $\E [M_x] = \E [M_y] = m/2$. By a standard Chernoff bound:
\begin{align*}
\Pr [M_x \le m/4] \le 2e^{- (1/2)^2\cdot m/4} = O \left (\frac{1}{m+1} \right). 
\end{align*}
(Again writing $m+1$ instead of $m$ to cover the $m=0$ case).
An identical bound holds for $M_y$. Thus, by a  union bound, except with probability $O \left (\frac{1}{m+1} \right)$ both $M_x$ and $M_y$  are $\ge m/4$. Applying Corollary \ref{cor:hitting} we have:
\begin{align*}
\Pr \left [ \mathcal{C}\right ] &= \Pr \left [ \mathcal{C} | M_x \ge m/4\text{ and }  M_y \ge m/4 \right ] \cdot \Pr \left [M_x \ge m/4\text{ and } M_y \ge m/4 \right ]
\\ &\hspace{5em}+ \Pr \left [ \mathcal{C} | M_x < m/4\text{ or } M_y < m/4 \right ] \cdot \Pr \left [M_x < m/4\text{ or } M_y < m/4 \right ]\\
&\le \Pr \left [ \mathcal{C} | M_x \ge m/4\text{ and }  M_y \ge m/4 \right ]  + \Pr \left [M_x < m/4\text{ or } M_y < m/4 \right ]\\
&= O \left (\frac{1}{\sqrt{(m/4+1)(m/4+1)}} + \frac{1}{A} + \frac{1}{m+1} \right )\\
&= O \left ( \frac{1}{m+1} + \frac{1}{A} \right  ),
\end{align*}
which gives the lemma.
\end{proof}

We note that Lemma \ref{lem:hitting2} immediately gives a bound of $O \left ( \frac{1}{m+1} + \frac{1}{A} \right  )$ on the probability that a single random walk returns to its origin (equalizes) after $m$ steps. In fact a slightly stronger bound can be shown in this special case:
\begin{corollary}[Equalization Probability Bound]\label{equalizationProbBound} Consider agent $a_1$ randomly walking on a two-dimensional torus of dimensions $\sqrt{A} \times \sqrt{A}$. If $a_1$ is located at position $p$ after round $r$, for any even $m \ge 0$, the probability that $a_1$ is again at position $p$ after round $r+m$ is $\Theta \left (\frac{1}{m+1} \right ) + O \left (\frac{1}{A} \right )$. For odd $m$ the probability is $0$.
\end{corollary}
\begin{proof}
The corollary has a $\Theta  \left ( \frac{1}{m+1} \right )$ bound instead of the $O \left ( \frac{1}{m+1} \right )$ bound  which would be given by directly applying Lemma \ref{lem:hitting2}. To obtain the stronger  bound, simply note that when bounding the equalization probability, we have $\delta_x = 0$ (where $\delta_x$  is as defined in the proof of Lemma \ref{lem:hitting}). As long as $m_x$ is even, the bound in Claim \ref{clm:c1} becomes $\Theta \left ( \frac{1}{\sqrt{m_x + 1}}\right)$. The remainder of the proof goes through unchanged, after  noting that if $m$ is even, $M_x$ and $M_y$ are both even with $\Theta(1)$ probability.

\end{proof}

We finally use Lemma \ref{lem:hitting2} to directly bound the probability of two random walks re-colliding after $m$ steps, conditioned on the path of one of these walks. This yields our re-collision probability bound, Lemma \ref{collideprobbound}.

\begin{proof}[Proof of Lemma \ref{collideprobbound}]
Consider any $m$-step path $w$. Let $(\xend,\yend)$ be the last position in $w$. Then, conditioned on $\mathcal{W} = w$, $a_2$ is at $(\xend,\yend)$ in round $r+m$. Thus, conditioned on $\mathcal{W} = w$, a re-collision occurs in this round if and only if $a_1$ is also located at $(\xend,\yend)$. Lemma \ref{lem:hitting2} gives a  bound on this probability, which yields the lemma.
\end{proof}

\subsection{Collision Moment Bound}

With Lemma \ref{collideprobbound} in hand, we can prove our collision moment bound, which we will use to show that the number of collisions an agent sees concentrates strongly around its expectation. Our moment  bound is:
\begin{lemma}[Collision Moment Bound]\label{per_agent_moments} Let $\mathcal{W}$ be the $t$-step random walk that an agent executing Algorithm \ref{random_walk_sampling} takes.
For $j \in [1,...,n]$, let $\bar c_j \eqdef c_j - \E [c_j | \mathcal{W}]$ and assume $t \le A$. There is some fixed constant $w$ such that for any integer $k \ge 1$, $$\E \left [ \bar c_j^k | \mathcal{W} \right] \le  \frac{t w^k}{A} \cdot k! \log^{k}(2t).$$
\end{lemma}
Note that, again, we condition the random walk taken by one of the agents.
When $k =2$, Lemma \ref{per_agent_moments} gives a bound on the \emph{variance} of $c_j$, which can be used to show that $c_j$ falls close to its mean with good probability. By bounding the $k^{th}$ moment $\E  [ \bar c_j^k | \mathcal{W}] $ for all $k$, we are able to show even stronger concentration results.
Our proof of Lemma \ref{per_agent_moments} breaks down into the following steps:
\begin{enumerate}
\item In Lemma \ref{original_collision_probability} we bound the probability that  an agent collides with any other particular agent $j$ at least once during the execution of Algorithm \ref{random_walk_sampling} (i.e., that $c_j \ge 1$), using a simple linearity  of expectation argument.
\item In Claims \ref{clm:per_agent_moments0} and \ref{clm:per_agent_moments1} we bound the $k^{th}$ moment of $c_j$ for all $k \ge 1$ conditioned on $c_j = 0$ (i.e., there is no collision with agent $a_j$) and on $c_j \ge 1$ (i.e., there is at least one collision with agent $a_j$. The $c_j \ge 1$ case uses the re-collision probability bound of Lemma \ref{collideprobbound}.
\item We combine the above results to give our final moment bound for each $c_j$, yielding Lemma \ref{per_agent_moments}.
\end{enumerate}

\begin{lemma}[First Collision Probability]\label{original_collision_probability} Let $\mathcal{W}$ be the $t$-step random walk that an agent executing Algorithm \ref{random_walk_sampling} takes. For all $j \in [1,...,n]$, 
$$\Pr\left [ c_j \ge 1 | \mathcal{W} \right ] \le \frac{t}{A}.$$
\end{lemma}
\begin{proof} Using the fact that $c_j$ is identically distributed for all $j$, and applying Lemma \ref{estimate_expectation},
\begin{align*}
\E[ \tilde d | \mathcal{W}] = d = \frac{1}{t} \cdot \E \left [\sum_{i=1}^n c_i \Big | \mathcal{W} \right] = \frac{n}{t} \cdot \E [c_j | \mathcal{W}] &= \frac{n}{t} \cdot  \Pr\left [ c_j  \ge 1  | \mathcal{W} \right] \cdot \E [ c_j | \mathcal{W}, c_j \ge 1 ] \\
\frac{n}{A} &= \frac{n}{t} \cdot   \Pr\left [ c_j  \ge 1  | \mathcal{W} \right] \cdot \E [ c_j | \mathcal{W}, c_j \ge 1 ].
\end{align*}
Rearranging and noting that $\E [ c_j | \mathcal{W}, c_j \ge 1 ] \ge 1$ gives: 
\begin{align}\label{exp_to_prob}
\Pr\left [ c_j  \ge 1 | \mathcal{W}\right ] = \frac{t}{A \cdot \E [ c_j  | \mathcal{W},c_j  \ge 1  ] } \le \frac{t}{A}.
\end{align}
%
%
\end{proof}

We next give a simple bound on the moments of $c_j$ conditioned on $c_j = 0$.
\begin{claim}\label{clm:per_agent_moments0}
Let $\mathcal{W}$ be the $t$-step random walk that an agent executing Algorithm \ref{random_walk_sampling} takes.
For $j \in [1,...,n]$, let $\bar c_j \eqdef c_j - \E [c_j | \mathcal{W}]$ and assume $t \le A$. For any integer $k \ge 1$, $$\E \left [ \bar c_j^k | \mathcal{W}, c_j = 0 \right] \le  \frac{t}{A}.$$
\end{claim}
\begin{proof}
By the argument given in Lemma \ref{estimate_expectation}, $\E[c_j | \mathcal{W}] = \sum_{r=1}^t \E[ c_j(t) | \mathcal{W}] = \frac{t}{A}$. We thus have:
$\E \left [ \bar c_j^k | \mathcal{W},c_j = 0\right] = \left (0 - \E[c_j | \mathcal{W}]\right)^k \le (t/A)^k$. Further since $t \le A$ by assumption, $t/A \le 1$ and we can loosely bound $(t/A)^k \le \frac{t}{A}$ for all $k \ge 1$, giving the claim.
\end{proof}

We next use the re-collision probability bound of Lemma \ref{collideprobbound} to bound the moments of $c_j$ conditioned on $c_j \ge 1$. 

\begin{claim}\label{clm:per_agent_moments1} Let $\mathcal{W}$ be the $t$-step random walk that an agent executing Algorithm \ref{random_walk_sampling} takes.
For $j \in [1,...,n]$, let $\bar c_j \eqdef c_j - \E [c_j | \mathcal{W}]$ and assume $t \le A$. There is some fixed constant $w$ such that for any integer $k \ge 1$, $$\E \left [ \bar c_j^k | \mathcal{W},c_j \ge 1 \right] \le  w^k \cdot k! \log^{k}(2t).$$
\end{claim}
\begin{proof}
Since $\E [c_j|\mathcal{W}] = \frac{t}{A} \le 1$, we have $\E \left [ \bar c_j^k | \mathcal{W},c_j \ge 1\right]  \le \E \left [ c_j^k | \mathcal{W},c_j \ge 1\right] $.
So to prove the lemma, it just suffices to show that $\E \left [ c_j^k | \mathcal{W},c_j \ge 1\right] \le k! w^k \log^k(2t)$ for some $w$ and all $k \ge 1$.

Let $t'$ be the first time in which there is a collision with agent $j$, and $t' = 1$ if there is no such collision.  
We split $c_j$ over rounds as
$
c_j = \sum_{r=t'}^t c_j(r) \le \sum_{r=t'}^{t'+t-1} c_j(r).
$
Where we simply  define $c_j(r) = 0$ for any  $r > t$.
To simplify notation we relabel round $t'$ round $1$ and so round $t'+t-1$ becomes round $t$. After this relabeling, conditioned on $c_j \ge 1$, we have $c_j(1) = 1$. 
 Expanding $c_j^k$ out fully using the summation:
\begin{align*}
\E \left [c_j^k | \mathcal{W},c_j \ge 1\right ] &= \E \left [ \sum_{r_1=1}^t \sum_{r_2=1}^t ... \sum_{r_k=1}^t c_{j}(r_1)c_{j}(r_2)...c_{j}(r_k) \Big | \mathcal{W},c_j \ge 1 \right ]\\
&= \sum_{r_1=1}^t \sum_{r_2=1}^t ... \sum_{r_k=1}^t \E \left [c_{j}(r_1)c_{j}(r_2)...c_{j}(r_k) | \mathcal{W},c_j \ge 1 \right ].
\end{align*}
$ \E \left [c_j(r_1)c_j(r_2)...c_j(r_k) | \mathcal{W},c_j \ge 1 \right ]$ is just the probability that the two agents collide in each of  rounds $r_1,r_2,...,r_k$, conditioned on the walk $\mathcal{W}$ and that $c_j(1) = 1$.
Assume without loss of generality that $r_1 \le r_2 \le ... \le r_k$. By Lemma \ref{collideprobbound} and the fact that $c_j(1) = 1$, for some fixed $w$ we can bound this probability
$
\le \frac{w^k}{r_1(r_2-r_1+1)(r_3-r_2+1)...(r_k-r_{k-1}+1)}.
$ Here we use the assumption that $t\le A$ so the $O\left ( \frac{1}{A} \right)$ term is absorbed into the $O \left ( \frac{1}{m+1} \right)$ term in Lemma \ref{collideprobbound}.
We then rewrite, by linearity of expectation:
\small
\begin{align*}
\E \left [c_j^k | \mathcal{W},c_j \ge 1 \right ] &\le k! \sum_{r_1=1}^t ...\sum_{r_k=r_{k-1}}^t \frac{w^k}{r_1(r_2-r_1+1)...(r_k-r_{k-1}+1)}.
\end{align*}
\normalsize
The $k!$ comes from the fact that in this sum we have only ordered $k$-tuples and so need to multiple by $k!$ to account for the fact that the original sum is over unordered $k$-tuples.
We can bound:
\begin{align*}
\sum_{r_k=r_{k-1}}^t \frac{1}{r_k-r_{k-1}+1} = 1 + \frac{1}{2} + ... + \frac{1}{t} = O( \log 2t),
\end{align*}
so rearranging the sum and simplifying gives:
\begin{align*}
\E \left [c_j^k | \mathcal{W},c_j \ge 1 \right ]  &\le k! w^k \sum_{r_1=1}^t \frac{1}{r_1}\sum_{r_2=r_1}^t \frac{1}{r_2-r_1+1}... \sum_{r_k=r_{k-1}}^t \frac{1}{r_k-r_{k-1}+1}\\
&\le k! w^k \sum_{r_1=1}^t ...\sum_{r_{k-1}=r_{k-2}}^t \frac{1}{r_{k-2}-r_{k-1}+1} \cdot  O(\log 2t).
\end{align*}
We repeat this argument for each level of summation replacing $\sum_{r_i=r_{i-1}}^t \frac{1}{r_i-r_{i-1}+1}$ with $O(\log 2t)$.
Iterating through the $k$ levels gives 
$$\E \left [c_j^k | \mathcal{W},c_j \ge 1 \right ] \le k! w^k \log^k 2t,$$
after $w$ is adjusted using the constant in the $O(\log 2t)$ term, establishing the claim.
\end{proof}

Finally, we combine claims \ref{clm:per_agent_moments0} and \ref{clm:per_agent_moments1} with the first collision probability bound of Lemma \ref{original_collision_probability} to prove our main moment bound, Lemma \ref{per_agent_moments}.


\begin{proof}[Proof of Lemma \ref{per_agent_moments}]
We expand: $$\E  [ \bar c_j^k | \mathcal{W}] = \Pr[c_j \ge 1 | \mathcal{W}] \cdot \E  [ \bar c_j^k | \mathcal{W}, c_j \ge 1] + \Pr[c_j = 0|\mathcal{W}] \cdot \E  [ \bar c_j^k | \mathcal{W}, c_j = 0].$$ By Lemma \ref{original_collision_probability}:
\begin{align*}
\E \left [ \bar c_j^k | \mathcal{W} \right] \le \frac{t}{A}\cdot \E \left [ \bar c_j^k | \mathcal{W}, c_j \ge 1\right] + \E \left [ \bar c_j^k | \mathcal{W}, c_j = 0\right].
\end{align*}
Plugging in the bounds of Claims \ref{clm:per_agent_moments0} and \ref{clm:per_agent_moments1} we then have, for some fixed $w$ and all $k \ge 1$,
\begin{align*}
\E \left [ \bar c_j^k | \mathcal{W} \right] &\le \frac{t w^k}{A} k! \log^k(2t) + \frac{t}{A}\\
&\le \frac{t (w+1)^k}{A} k! \log^k(2t),
\end{align*}
giving the lemma.
\end{proof}

As with Lemma \ref{collideprobbound}, the techniques used in Lemma \ref{per_agent_moments} can be applied to bounding the moments of the number of equalizations of a single random walk. We give two bounds that may be of independent interest. Note that the first bound is slightly tighter (by a $\log 2t$ factor) than what would be obtained simply by replacing the use of Lemma \ref{collideprobbound} with Corollary  \ref{equalizationProbBound} in the proof of Claim \ref{clm:per_agent_moments1}.

\begin{corollary}[Random Walk Visits Moment Bound]\label{visitMomentBound} Consider an agent $a_1$ randomly walking on a two-dimensional $\sqrt{A} \times \sqrt{A}$ torus that is initially located at a uniformly random location and takes $t \le A$ steps. Let $c_j$ be the number of times that $a_1$ visits node $j$. There exists a fixed constant $w$ such that for all $j \in [1,...A]$ and all $k \ge 1$,
\begin{align*}
\E \left [ \bar c_j^k \right] \le \frac{tw^k}{A} \cdot k! \log^{k-1}(2t).
\end{align*} 
\end{corollary}
\begin{proof}
We show that $\Pr[c_j \ge 1] = O \left ( \frac{t}{A \log 2t} \right)$, strengthening Lemma \ref{original_collision_probability} by a $\log 2t$ factor. Combining this stronger bound with 
Claims \ref{clm:per_agent_moments0} and  \ref{clm:per_agent_moments1} and following the proof of Lemma \ref{per_agent_moments} gives the result.

Let $c(r)$ be $1$ if the agent visits node $j$ in round $r$, and $0$ otherwise.
Due to the initial uniform distribution of the agent, by linearity of expectation: 
$$\E[c] = \sum_{i=1}^t  \E[c(r)] = \frac{t}{A}.$$
As in the proof of Lemma \ref{original_collision_probability}, we can rewrite this expectation as:
\begin{align*}
\E[ c] = \frac{t}{A} = \Pr[c \ge 1] \cdot \E[c | c \ge 1].
\end{align*}
To compute $\E [ c | c  \ge 1  ] $, we use  Corollary \ref{equalizationProbBound} and linearity of expectation. Since $t \le A$, the $O\left (\frac{1}{A}\right)$ term in Corollary \ref{equalizationProbBound} is absorbed into the $\Theta \left ( \frac{1}{m+1}\right)$. Let $r \le t$ be the first round that the agent visits node $j$. Then:
\begin{align}\label{expectation_cj}
\E [ c  | c  \ge 1  ]  = \sum_{m = 0}^{t-r} \Theta \left (\frac{1}{m+1} \right )= \Theta \left ( \log (2(t-r+1)) \right ).
\end{align}
Further, the probability of the \emph{first} visit to node $j$ is in a given round can only decrease as the round number increases. So, at least $1/2$ of the time that $c \ge 1$, there is a visit in the first $t/2$ rounds (Note that we can assume $t \ge 2$ since if $t =1$ we already have $\E[c | c \ge 1] = 1$). So, overall, by \eqref{expectation_cj}, $\E [ c  | c  \ge 1  ]  = \Theta \left (\log (2(t-t/2+1)) \right ) = \Theta \left (\log 2t \right )$. Using \eqref{exp_to_prob}, $\Pr\left [ c  \ge 1 \right ] =O\left ( \frac{t}{A  \log 2t } \right )$, completing the proof.
\end{proof}
We have a similar bound on the number of returns to the agent's starting node.


\begin{corollary}[Equalization Moment Bound]\label{equalizationMomentBound} Consider an agent $a_1$ randomly walking on a two-dimensional $\sqrt{A} \times \sqrt{A}$ torus. If $a_1$ takes $t \le A$ steps and $c$ is the number of times it returns to its starting position (the number of equalizations), there exists a fixed constant $w$ such that for all $k \ge 1$, $$\E \left [ \bar c^k \right] \le k! w^k \log^{k} (2t).$$
\end{corollary}
\begin{proof}
This follows an identical proof to that of the moment bound given in Claim \ref{clm:per_agent_moments1} for the number of collisions between two agents that are assumed to collide at least once: $\E \left  [c_j^k | \mathcal{W}, c_j \ge 1 \right ] \le k!w^k\log^k(2t)$. We simply replace the application of Lemma \ref{collideprobbound} with Corollary \ref{equalizationProbBound}.
\end{proof}

\subsection{Correctness of Encounter Rate With High Probability}\label{sec:removeCond}

\label{concentration_section}


Armed with Lemma \ref{per_agent_moments} we can finally show that $\sum_{j=1}^n c_j$ concentrates strongly about its expectation. Since $\tilde d = \frac{1}{t} \sum_{j=1}^n c_j$, this is enough to prove the accuracy of encounter-rate-based density estimation (Algorithm \ref{random_walk_sampling}).
We first use Lemma \ref{per_agent_moments} to give a standard `sub-exponential' bound on the sum $\sum_{j=1}^n c_j$ (Corollary \ref{subExpCorrollary} below). It is in this step that we use that each $c_j$ is independent conditioned on executing agent's random walk $\mathcal{W}$. We correspondingly remove this conditioning, giving an unconditioned bound.

\begin{corollary}[Sub-exponential condition]\label{subExpCorrollary}  Assuming $t\le A$, for some $b = \Theta(\log 2t)$, $\sigma^2  = \Theta(td \log^2 2t)$, and any $\lambda$  with $|\lambda| \le 1/b$ we have:
$$
\E \left [ e^{\lambda \left (\sum_{j=1}^n c_j - \E\left [\sum_{j=1}^n c_j  \right ]\right )}  \right ] \le e^{\frac{\lambda^2 \sigma^2}{2-2b|\lambda|}}.
$$
\end{corollary}
\begin{proof}
Let $\mathcal{W}$ be the $t$-step random walk that an agent executing Algorithm \ref{random_walk_sampling} takes.
By Lemma \ref{per_agent_moments}, there exists some constant $w$ such that for $\sigma^2 = \frac{w^2t\log^2 2t}{A}$ and $b = w\log 2t$, $\bar c_j \eqdef c_j - \E[ c_j | \mathcal{W}]$ satisfies:
\begin{align*} 
\E \left [ \bar c_j^k | \mathcal{W}\right ] \le \frac{1}{2}k! \sigma^2 b^{k-2}.
\end{align*}
By Proposition 2.3 of \cite{wainwright}, this gives the sub-exponential moment bound: for any $\lambda$ with $|\lambda| \le 1/b$, 
$$\E[e^{\lambda \bar c_j} | \mathcal{W}] \le e^{\frac{\lambda^2 \sigma^2}{2-2b|\lambda|}}.$$
Since each $c_j$ is independent conditioned on the walk $\mathcal{W}$  this gives:
\begin{align*}
\E \left [ e^{\lambda \left (\sum_{j=1}^n c_j - \E\left [\sum_{j=1}^n c_j  |\mathcal{W}\right ]\right )} | \mathcal{W} \right ] = \E \left [\prod_{j=1}^n e^{\lambda \bar c_j} \Big | \mathcal{W} \right ] = \prod_{j=1}^n  \E \left  [e^{\lambda \bar c_j} | \mathcal{W} \right  ] \le e^{\frac{n \lambda^2 \sigma^2}{2-2b|\lambda|}}.
\end{align*}
The lemma follows by replacing $\sigma^2$ with $n\sigma^2 = \Theta(td\log^2 2t)$. Further, the above bound holds for all $\mathcal{W}$ and $\E[c_j |\mathcal{W}] = \E[c_j]$ (see Corollary \ref{cor:expectation}). So, by the law of iterated expectation, we remove the conditioning on $\mathcal{W}$:
\begin{align*}
\E \left [ e^{\lambda \left (\sum_{j=1}^n c_j - \E\left [\sum_{j=1}^n c_j  \right ]\right )}\right ] = \E \left [\E \left [ e^{\lambda \left (\sum_{j=1}^n c_j - \E\left [\sum_{j=1}^n c_j |\mathcal{W} \right ]\right )} \Big | \mathcal{W} \right ] \right ] \le e^{\frac{n \lambda^2 \sigma^2}{2-2b|\lambda|}}.
\end{align*}
\end{proof}

We will employ a concentration bound for random variables satisfying such a sub-exponential condition: 

\begin{lemma}[Proposition 2.3 of \cite{wainwright}]\label{tail_bound}
Suppose that $X$ satisfies $\E \left [e^{\lambda(X-\E [X])}\right ] \le e^{\frac{\lambda^2 \sigma^2}{2-2b|\lambda|}}$ for any $\lambda$ with $|\lambda| \le 1/b$. Then for any $\Delta \ge 0$,
$$
\Pr \left [\left |X - \E [X] \right  | \ge \Delta \right] \le 2e^{-\frac{\Delta^2}{2(\sigma^2+b\Delta)}}.
$$
\end{lemma}
\begin{proof}  This bound is given in Proposition 2.3 of \cite{wainwright}.
We include a full proof for completeness.
We have $\Pr \left [\left |X - \E [X]\right| \ge \Delta \right] = \Pr \left [\left (X - \E [X]\right) \ge \Delta \right] + \Pr \left [\left (X - \E [X]\right) \le -\Delta \right]$. We can bound these terms similarly. Focusing on the first, for any positive $\lambda$,
\begin{align*}
\Pr \left[\left (X - \E [X]\right) \ge \Delta \right]  = \Pr \left [e^{\lambda (X - \E [X])} \ge e^{\lambda \Delta} \right ].
\end{align*}
By Markov's inequality and our moment bound, for any $\lambda$ with $|\lambda| \le 1/b$:
\begin{align*}
\Pr \left [e^{\lambda (X - \E [X])} \ge e^{\lambda \Delta} \right ] \le \E \left [e^{\lambda |X - \E [X]|} \right ]\cdot e^{-\lambda \Delta} \le e^{\left(\frac{\lambda^2 \sigma^2}{2-2b|\lambda|} - \lambda \Delta\right)}.
\end{align*}
We can set  $\lambda = \frac{\Delta}{\sigma^2+b\Delta}$ and calculate:
\begin{align*}
\frac{\lambda^2 \sigma^2}{2-2b|\lambda|} - \lambda \Delta = \frac{\Delta^2 \sigma^2}{( \sigma^2+b\Delta)^2 \cdot \frac{2\sigma^2}{ \sigma^2 + b\Delta}} - \frac{\Delta^2}{\sigma^2+b\Delta} = -\frac{\Delta^2}{2( \sigma^2+b\Delta)}.
\end{align*}
Which gives $\Pr \left [\left (X - \E [X]\right) \ge \Delta  \right]  \le e^{-\frac{\Delta^2}{2(\sigma^2+ b\Delta )}}$. We can bound $\Pr \left[\left(X - \E [X]\right) \le -\Delta\right ]$ in the same way by setting $\lambda = -\frac{\Delta}{\Delta  b+\sigma^2}$, giving the Lemma.
\end{proof}


We conclude by proving our main theorem on the accuracy of random-walk-based density estimation:
\begin{proof}[Proof of Theorem \ref{naturalAlgoThm}]
In Algorithm \ref{random_walk_sampling}, $\tilde d$ is set to $\frac{1}{t}\sum_{j=1}^n c_j$. So the probability that $\tilde d$ falls within an $\epsilon$ multiplicative factor of its mean is the same as the probability that $\sum_{j=1}^n c_j$ falls within an $\epsilon$ multiplicative factor of its mean, which is equal to $t \E[\tilde d] = td$ by Corollary \ref{cor:expectation}. By Corollary \ref{subExpCorrollary} and Lemma \ref{tail_bound}:
\begin{align*}
\delta &\eqdef \Pr \left [ \left |\sum_{j=1}^n c_j - \E \left [\sum_{j=1}^n c_j \right ] \right |\ge \epsilon \E \left [\sum_{j=1}^n c_j \right ] \right]\\
&=\Pr \left [ \left |\sum_{j=1}^n c_j - td \right |\ge \epsilon td \right]
\le 2e^{\Theta \left (-\frac{\epsilon^2 t^2 d^2}{2(td\log^2 2t + \epsilon td \log 2t)} \right )} = 2e^{\Theta \left (-\frac{\epsilon^2 t d}{\log^2 2t} \right )},
\end{align*}
where the last equality follows since we restricting $\epsilon \le 1$. The above thus gives
$
\frac{\epsilon^2td}{\log^2 2t} = O \left ( \log(1/\delta) \right )$ and so
$\epsilon \le  c_1 \cdot \sqrt{\frac{\log(1/\delta)}{td}} \cdot  \log 2t$ for some fixed constant $c_1$. This gives the first claim of the theorem: for any $\delta > 0$, with probability $1-\delta$, 
\begin{align}\label{eq:endgame}
\tilde d \in [(1-\epsilon)d,(1+\epsilon)d]\text{ for }\epsilon \le c_1 \cdot \sqrt{\frac{\log(1/\delta)}{td}} \cdot  \log 2t.
\end{align}

Given any fixed $\epsilon,\delta  \in (0,1)$ 
we can also rearrange by solving \eqref{eq:endgame} for $t$ which gives:
$\frac{t}{\log^2 2t} \ge \frac{c_1^2 \log(1/\delta)}{\epsilon^2 d}.$
If we set, for some constant $c_3$, $$t = \frac{c_3\cdot c_1^2\log(1/\delta)[\log \log(1/\delta) + \log(1/d\epsilon)]^2}{d\epsilon^2},$$ we have $\log 2t = \log(2c_3) + c_4[\log \log(1/\delta) + \log(1/d\epsilon)]$ for some fixed constant $c_4$ which is independent of $c_3$. Thus, setting $c_3$ large enough gives $\frac{t}{\log^2 2t} \ge \frac{c_1^2 \log(1/\delta)}{\epsilon^2 d}$.

Setting $c_2 = c_3 \cdot c_1^2$, this yields the second claim of the theorem:
\begin{align*}
\text{for any }\epsilon,\delta  \in (0,1),\text{ if }A \ge t \ge  \frac{c_2\log(1/\delta)[\log \log(1/\delta) + \log(1/d\epsilon)]^2}{d\epsilon^2},
\end{align*} 
then $\tilde d \in [(1-\epsilon)d,(1+\epsilon)d]$ with probability $\ge 1-\delta$.
\end{proof}

\section{Extensions to Other Regular Topologies}\label{generalizations}

We now discuss extensions of our results to a broader set of graph topologies, demonstrating the generality of our local mixing analysis. We illustrate divergence between local and global mixing properties, which can have significant effects on random-walk-based algorithms

\subsection{From Re-collision Bounds to Accurate Density Estimation}
 
Our proofs for the two-dimensional torus are largely independent of graph structure, using just a re-collision probability bound (Lemma \ref{collideprobbound}) and the regularity (uniform node degrees) of the grid, so agents remain uniformly distributed on the nodes in each round (see for example, Lemma \ref{estimate_expectation}). Hence, extending our results to other regular graphs primarily involves obtaining re-collision probability bounds for these graphs.

We consider agents on a graph with $A$ nodes that execute analogously to Algorithm \ref{random_walk_sampling}, stepping to a random neighbor in each round. Again, we focus on the multi-agent case but similar bounds (resembling Corollaries \ref{visitMomentBound} and \ref{equalizationMomentBound}) hold for a single random walk.
We start with a lemma that gives density estimation accuracy in terms of re-collision probability. This is a direct generalization of our grid analysis.

\begin{lemma}[Re-collision Probability to Density Estimation Accuracy]\label{generalBound}
Consider a regular graph (uniform nodes degrees) with $A$ nodes along with two agents $a_1$ and $a_2$ randomly walking on this graph. Assume that $a_1$ and $a_2$ collide in round $r$. Suppose that there exists some function $\beta(m)$ such that, for any $0 \le m \le t$, letting $\mathcal{W}$ be the random $m$-step path taken by $a_2$ in rounds $r+1,...,r+m$, and $\mathcal{C}$ be the event that $a_1$ and $a_2$ re-collide in round $r+m$:
$$Pr[\mathcal{C} | \mathcal{W}] = O(\beta(m)).$$
Let $B(t) \eqdef \sum_{m = 0}^{t} \beta(m)$. After running for $t \le A$ steps, Algorithm \ref{random_walk_sampling} returns $\tilde d$ such that, for any $\delta > 0$, with probability $\ge 1-\delta$, $$\tilde d \in [(1-\epsilon)d, (1+\epsilon)d]\text{ for }\epsilon = O\left ( \sqrt{\frac{\log(1/\delta)}{td}}\cdot   B(t)\right ).$$
\end{lemma}
Note that in the special case of the two-dimensional torus, by Lemma \ref{collideprobbound}, we can set $\beta(m) = 1/(m+1)$ and hence $B(t) = O(\log 2t)$, recovering Theorem \ref{naturalAlgoThm}.
\begin{proof}
Let $\mathcal{W}$ denote the $t$-step random walk of an agent executing Algorithm \ref{random_walk_sampling}.
$\E [\tilde d | \mathcal{W}] = d$ (Lemma \ref{estimate_expectation}) still holds as the regularity of the graph ensures that agents remain uniformly distributed on the nodes in every round (the stable distribution of any regular graph is the uniform distribution). 

Further,
 following the moment calculations in Claim \ref{clm:per_agent_moments1}, $\E [ c_j^k | \mathcal{W}, c_j \ge 1] \le k! w^k B(t)^k$ for some constant $w$.  Claim \ref{clm:per_agent_moments0} and Lemma \ref{original_collision_probability} still hold, giving that the bound of Lemma \ref{per_agent_moments} holds unchanged:
\begin{align*}
\E[\bar c_j^k | \mathcal{W}] \le \frac{tw^k}{A}\cdot k! B(t)^{k}. 
\end{align*}
As in Corollary \ref{subExpCorrollary}, this gives that $\sum_{j=1}^n c_j$ satisfies the sub-exponential condition
\begin{align*}\E \left [ e^{\lambda \left (\sum_{j=1}^n c_j - \E\left [\sum_{j=1}^n c_j  \right ]\right )}  \right ] \le e^{\frac{\lambda^2 \sigma^2}{2-2b|\lambda|}}
\end{align*}
for $b = \Theta(B(t))$ and $\sigma^2 = \Theta \left (tdB(t)^2\right )$.
Plugging into Lemma \ref{tail_bound} gives
$\frac{\epsilon^2 td}{B(t)^2} = O(\log(1/\delta))$. Rearranging yields the result.
\end{proof}

\subsection{Density Estimation on the Ring}\label{sec:ring}

We first  consider the ring, where  the following re-collision probability bound  holds:

\begin{lemma}[Re-collision Probability Bound -- Ring]\label{ringBound}
Consider two agents $a_1$ and $a_2$ randomly walking on a one-dimensional torus (a ring) with $A$ nodes. Assume that $a_1$ and $a_2$ collide in round $r$. For any $m \ge 0$, let $\mathcal{W}$ be the $m$-step random walk performed by $a_2$ in rounds $r+1,...,r+m$. Let $\mathcal{C}$ be the event that $a_1$ and $a_2$ collide again in round $r+m$. We have:
\begin{align*}
\Pr[\mathcal{C}|\mathcal{W}] = O\left (\frac{1}{\sqrt{m+1}} + \frac{1}{A} \right ).
\end{align*}
\end{lemma}
\begin{proof}
This bound holds via an essentially identical proof to the bound on $\mathcal{C}_x$ given in Claim \ref{lem:hitting}. An $m$-step random walk on a line ends at any  position with probability $O(1/\sqrt{m+1})$. On a ring with $A$ nodes the slightly weaker bound of $O\left (\frac{1}{\sqrt{m+1}} + \frac{1}{A} \right )$ holds.
\end{proof}

\paragraph{Density Estimation Bound:}
For $m \le A$, the $O \left (\frac{1}{A}\right )$ term in Lemma \ref{ringBound} is absorbed into the $O\left (\frac{1}{\sqrt{m+1}} \right )$ and one can show that $\sum_{m=0}^t 1/\sqrt{m+1} = \Theta (\sqrt{t})$. Plugging into Lemma \ref{generalBound}, we obtain 
$$\epsilon = O\left ( \sqrt{\frac{\log(1/\delta)}{td}}  \cdot  \sqrt{t} \right )= O\left (\sqrt{\frac{\log(1/\delta)}{d}}\right ).$$ Note that  this bound does not depend on $t$.  That  is, since 
local mixing on the ring is much worse than on the torus, our general technique is not strong enough  to show the convergence  of  random-walk-based density estimation.
We do note that 
we can give a density estimation accuracy bound by using an alternative analysis which bounds the variance of the collision count $c = \sum c_j$ and applies Chebyshev's inequality. 
This technique is very similar to what is used in our network size estimation bounds in Section \ref{sec:size}. 
A similar analysis can be applied to the two-dimensional torus and the other graphs we consider, giving worse dependence on the failure probability $\delta$, but slightly improved dependence on other parameters.
\begin{theorem}[Alternative Accuracy Bound -- Ring]\label{thm:ringCheby} After running for $t$ rounds, assuming $t \le A^2$, an agent executing Algorithm \ref{random_walk_sampling} on a ring with $A$ nodes returns $\tilde d$ such that, for any $\delta > 0$, with probability $\ge 1-\delta$, 
$$\tilde d \in [(1-\epsilon ) d,  (1+\epsilon) d ]\text{ for }\epsilon =O \left (\sqrt{\frac{1}{t^{1/2}d \delta}} \right ).$$
This implies that, for any $\epsilon, \delta \in (0,1)$ if $t = \Omega \left ( \frac{1}{(d\epsilon^2 \delta)^2} \right )$, $\tilde d \in [(1-\epsilon ) d,  (1+\epsilon) d ]$ with probability $\ge 1-\delta$.
\end{theorem}
Note that in Theorem \ref{thm:ringCheby}, $t$ is required to be quadratic in $\frac{1}{d\epsilon^2}$, rather than linear as in Theorem \ref{naturalAlgoThm}. The weakness of this bound is again due to the poor local mixing on the ring which means that, over $t$ steps, we expect to see $\Theta(\sqrt{t})$ rather than $\Theta(\log t)$ repeat collisions with every agent interacted with.
\begin{proof}
We first note that, if we do not condition on the path of either agent, the collision count $c_j$ with any agent $j$ is identically distributed to the number of times that a single random walk of length $2t$, initially placed uniformly at random, visits some location on the ring. We gave a bound on the moments of this visit count on the two-dimensional torus in Corollary \ref{visitMomentBound}.
A nearly identical analysis gives a similar bound on the ring. On the ring, the calculation of $\E[c|c\ge 1]$ in \eqref{expectation_cj} becomes 
\begin{align*}
\E[c_j|c_j \ge 1] = \sum_{m=0}^{t-r} \Theta \left (\frac{1}{\sqrt{m+1}} \right) = \Theta \left (\sqrt{t-r} \right )
\end{align*}
where $r$ is the round in which the first visit occurs. We again argue that $r < t/2$ with probability  $1/2$, which lets us show that $\Pr[c_j\ge 1] = O \left (\frac{\sqrt{t}}{A} \right )$ and overall:
\begin{align*}
\E [\bar c_j^k] \le \frac{t w^k}{A} k! \cdot t^{(k-1)/2}.
\end{align*}
where $\bar c_j = c_j - \E[c_j]$. 
 In particular, this gives the variance bound:
\begin{align*}
\E[\bar c_j^2] = O \left (\frac{t^{3/2}}{A} \right).
\end{align*}
We then have, recalling that $c = \sum_{i=1}^n c_j$ and letting $\bar c = c - \E[c]$,
\begin{align}\label{newVarBound}
\E\left [\bar c^2  \right] = \E \left [ \left (\sum_{i=1}^n \bar c_j \right)^2\right ] &= \sum_{i=1}^k \E \left [\bar c_j^2\right ] + \sum_{i=1}^n \sum_{j = i+1}^n 2 \E\left [\bar c_i \cdot \bar c_j\right ]\nonumber\\
&= O \left (\frac{nt^{3/2}}{A} \right ) + \sum_{i=1}^n \sum_{j = i+1}^n 2 \E[\bar c_i \cdot \bar c_j].
\end{align}
Let $\mathcal{W}$ be the path taken by the central agent $a$ and let $\mathcal{S}$ be the set of all possible instantiations of this path. $\bar c_i$ and $\bar c_j$ are independent conditioned on this path. So:
\begin{align*}
\E[\bar c_i \cdot \bar c_j] &= \sum_{w \in \mathcal{S}} \Pr [\mathcal{W} = w] \cdot \E [\bar c_i |\mathcal{W} = w ]\cdot \E[ \bar c_j | \mathcal{W} = w]\\
&= 0
\end{align*}
where we use that $\E[c_i] = td = \E[c_i | \mathcal{W} = w]$ for all $w$, and thus $\E [\bar c_i |\mathcal{W} = w ] = 0$.

Plugging back into \eqref{newVarBound} gives $\E \left [\bar c^2 \right ] = O \left (\frac{nt^{3/2}}{A}\right ) = O \left (d t^{3/2}\right)$. Recall that $\E[c] = dt$ and we thus have, via Chebyshev's inequality, 
\begin{align*}
\delta \eqdef \Pr \left [ |c - \E[c]| \ge \epsilon \E[c] \right ] =  O\left (\frac{dt^{3/2}}{(\epsilon dt)^2} \right) = O \left( \frac{1}{\epsilon^2 d \sqrt{t}} \right).
\end{align*}
Rearranging  gives $\epsilon = O \left ( \sqrt{\frac{1}{t^{1/2} d \delta }}\right)$ and so the theorem.
\end{proof}

\subsection{Density Estimation on $k$-Dimensional Tori}\label{sec:tori}

We next consider density  estimation on $k$-dimensional tori for general $k \ge 3$. As $k$ increases, local mixing becomes stronger, fewer re-collisions occur, and density estimation becomes easier. In fact, for any constant $k \ge 3$, although the torus still mixes slowly (with global mixing time on the order of $A^{2/k}$ \cite{aldous2002reversible}), density estimation is as accurate as on the complete graph! Throughout this section we assume that $k$ is a small constant and so hide multiplicative factors in $f(k)$ for any function $f$ in our asymptotic notation. We subscript the notation with $k$ to make this clear.

\begin{lemma}[Re-collision Probability Bound -- High-Dimensional Torus]
\label{collideprobboundd}
Consider two agents $a_1$ and $a_2$ randomly walking on a $k$-dimensional torus with $A$ nodes. Assume that $a_1$ and $a_2$ collide in round $r$. For any $m \ge 0$, let $\mathcal{W}$ be the $m$-step random walk performed by $a_2$ in rounds $r+1,...,r+m$. Let $\mathcal{C}$ be the event that $a_1$ and $a_2$ collide again in round $r+m$. We have:
\begin{align*}
\Pr[\mathcal{C}|\mathcal{W}] = O_k \left (\frac{1}{(m+1)^{k/2}} + \frac{1}{A}\right ).
\end{align*}
\end{lemma}
\begin{proof}
We closely follow the proof of Lemma \ref{collideprobbound}. 
In total, $a_1$ takes $m$ steps: $M_i$ in each dimension for $i \in [1,...,k]$.
Let $\mathcal{C}_i$ be the event that the agents have the same position in the $i^{th}$ dimension in round $r+m$. By the analysis of Lemma \ref{lem:hitting},
\begin{align*}
\Pr [\mathcal{C}_i | M_i = m_i ] = O \left (\frac{1}{\sqrt{m_i+1}} + \frac{1}{A^{1/k}} \right ).
\end{align*}
So, following the analysis of Corollary \ref{cor:hitting}, since movement in each of the $k$ directions is independent,
\begin{align}
\Pr [\mathcal{C} | M_1= m_1,...,M_k=m_k ] &= O \left (\frac{1}{\sqrt{m_1+1}} + \frac{1}{A^{1/k}} \right ) \cdot .... \cdot   O \left (\frac{1}{\sqrt{m_k+1}} + \frac{1}{A^{1/k}} \right )\label{hidimPrecond}.
\end{align}
We can then remove the conditioning on $M_1,...,M_k$ similarly to Lemma \ref{lem:hitting2}.
In expectation, $M_i = m/k$. So by a Chernoff bound, 
\begin{align*}
\Pr \left [M_i \le \frac{m}{2k}\right]  \le 2e^{-(1/2)^2\cdot m/3k} = O \left (\frac{1}{(m+1)^{k/2}} \right )
\end{align*}
again assuming $k$ is a small constant. Union bounding over all $k$ dimensions, we have $M_i \ge m/(2k)$ for all $i$ except with probability $O \left (\frac{1}{(m+1)^{k/2}} \right )$ and hence by \eqref{hidimPrecond}:
\begin{align*}
\Pr [\mathcal{C} ] = O \left (\frac{1}{(m+1)^{k/2}} \right ) + \left [ O\left (\frac{1}{\sqrt{m/(2k)+1}} +\frac{1}{A^{1/k}}\right ) \right ]^k = O_k \left (\frac{1}{(m+1)^{k/2}}+\frac{1}{A}\right ),
\end{align*}
giving the lemma (again, asymptotic notation hides multiplicative factors in $k$ since it is a constant).
\end{proof}

\paragraph{Density Estimation Bound:}
We can plug the bound of Lemma \ref{collideprobboundd} into Lemma \ref{generalBound}.
For $t \le A$ and $k \ge 3$, 
$$\sum_{m=0}^t  \left (\frac{1}{(m+1)^{k/2}} + \frac{1}{A} \right ) < 1 + \sum_{m=0}^\infty  \frac{1}{(m+1)^{k/2}} = O(1).$$ So we can set $B(t)=O_k(1)$ and have $\epsilon = O_k\left (\frac{\sqrt{\log(1/\delta)}}{td}\right )$. Rearranging, we require $t = O_k \left (\frac{\log(1/\delta)}{\epsilon^2 d}\right )$. This matches independent sampling up to constants and multiplicative factors in $k$.

\subsection{Density Estimation on Regular Expanders}
When a graph \emph{does} mix well globally, it mixes well locally. An obvious example is the complete graph, on which random-walk-based and independent-sampling-based density estimation are equivalent. We extend this intuition to any regular expander. An expander is a graph whose random walk matrix has its second eigenvalue bounded away from $1$, and so on which random walks mix quickly. Expanders are `well-connected' graphs with many applications, including in the design of robust communication networks \cite{bassalygo1973complexity} and efficient sampling schemes \cite{gillman1998chernoff}. 

\begin{lemma}[Re-collision Probability Bound -- Regular Expander]\label{collideprobboundexp} Let $G$ be a $k$-regular expander with $A$ nodes and adjacency matrix $\bv{M}$. Let $\bv{W} = \frac{1}{k}\cdot \bv{M}$ be its random walk matrix, with eigenvalues $\lambda_1 \geq \lambda_2 \geq ... \geq \lambda_A$. Let $\lambda = \max\{|\lambda_2|, |\lambda_A| \} < 1$. Consider two agents $a_1$ and $a_2$ randomly walking on $G$. Assume that $a_1$ and $a_2$ collide in round $r$. For any $m \ge 0$, let $\mathcal{W}$ be the $m$-step random walk performed by $a_2$ in rounds $r+1,...,r+m$. Let $\mathcal{C}$ be the event that $a_1$ and $a_2$ collide again in round $r+m$. We have:
\begin{align*}
\Pr[\mathcal{C}|\mathcal{W}] \le \lambda^m + 1/A.
\end{align*}
\end{lemma}
\begin{proof}
Suppose that $a_1$ and $a_2$ collide at node $i$ in round $r$.  
For any $m$-step path $w$ on $G$, conditioning on $\mathcal{W} = w$ fixes $a_2$'s position in round $r+m$, which we denote as node $j$. 
A re-collision occurs at around $r+m$ if and only if $a_1$ is also located at node $j$ in this round. The probability of this event can be bounded by
the maximum probability of $a_1$ being at any node $j$ in round $r+m$ after starting from node $i$. This equals $\max_{j \in \{1,...,|A|\}} (\bv{W}^m \bv{e}_i)_j$, which we can bound using 
the following lemma on how rapidly an expander random walk converges to its stable distribution:

\begin{lemma}[See \cite{lovasz1993random}]\label{lem:expandermixing} Let $G$ be a $k$-regular expander with $A$ nodes, adjacency matrix $\bv{M}$, and random walk matrix $\bv{W} = \frac{1}{k} \cdot \bv{M}$. Let $\lambda_1 \geq \lambda_2 \geq \ldots \geq \lambda_A$ be the eigenvalues of $\bv{W}$ and $\lambda = \max\{|\lambda_2|, |\lambda_A| \} < 1$. For each $1 \leq i,j \leq A$, $$\left|(\bv{W}^{m} \cdot \bv{e}_i)_j  - \frac{1}{A} \right| \leq \lambda^m.$$ \end{lemma}
\noindent This gives $\max_{j \in \{1,...,|A|\}} (\bv{W}^m \bv{e}_i)_j \le \lambda^m + \frac{1}{A}$, giving the Lemma.
%
\end{proof}

\paragraph{Density Estimation Bound:}
Again, we plug Lemma \ref{collideprobboundexp} in Lemma \ref{generalBound}, setting $B(t) = \sum_{m=0}^{t} \beta(m) \leq \frac{1}{1-\lambda} + t/A$. Assuming $t = O(A)$, 
\begin{align*}
\epsilon = O\left (\sqrt{\frac{\log(1/\delta)}{td}}\cdot \left (\frac{1}{1-\lambda} + \frac{t}{A}\right)^2\right )= O \left ( \sqrt{\frac{\log(1/\delta)}{td(1-\lambda)^2}}\right ).
\end{align*}
Rearranging, $t = O \left (\frac{\log(1/\delta)}{\epsilon^2d(1-\lambda)^2} \right )$, matching independent sampling up to a factor of $O(1/(1-\lambda)^2)$.

\subsection{Density Estimation $k$-Dimensional Hypercubes}
Finally, we give bounds for a $k$-dimensional hypercube. Such a graph has $A = 2^k$ vertices mapped to the elements of $\{\pm 1 \}^k$, with an edge between any two vertices that differ by hamming distance $1$. The hypercube is relatively fast mixing. Its adjacency matrix eigenvalues are $[-k, -k+2,...,k-2,k]$. Since it is bipartite, we can ignore the negative eigenvalues: to return to its origin, a random walk must take an even number of steps, so we need only need to consider the squared walk matrix $\bv{W}^2$, which has all positive eigenvalues. Applying Lemma \ref{collideprobboundexp} with $\lambda = \Theta(1-2/k) = \Theta(1-1/\log A)$, gives $t = O\left (\frac{\log(1/\delta)\log^2(A)}{\epsilon^2d} \right )$. However, it is possible to remove the dependence on $A$ via a more refined analysis -- while the global mixing time of the graph increases as $A$ grows, local mixing becomes stronger!


\begin{lemma}[Re-collision Probability Bound -- $k$-Dimensional Hypercube]
\label{collideprobboundhc} 
Consider two agents $a_1$ and $a_2$ randomly walking on a $k$-dimensional hypercube with $A=2^k$ vertices. Assume that $a_1$ and $a_2$ collide in round $r$. For any $m \ge 0$, let $\mathcal{W}$ be the $m$-step random walk performed by $a_2$ in rounds $r+1,...,r+m$. Let $\mathcal{C}$ be the event that $a_1$ and $a_2$ collide again in round $r+m$. We have:
\begin{align*}
\Pr[\mathcal{C}|\mathcal{W}] \le (9/10)^{m-1} + \frac{1}{\sqrt{A}}.
\end{align*}
\end{lemma}
\begin{proof}

A node of the hypercube can be represented as a $k$-bit string and each random walk step seen as choosing one of the bits uniformly at random and flipping it. Without loss of generality, assume that $a_1$ and $a_2$ collide at the node corresponding to bit string $[0,0,...,0]$  in round $r$. For any $m$-step path $w$, conditioning on $\mathcal{W} = w$ fixes $a_2$'s position in round $r+m$, which we denote as node $\pend$. We can write this position using its associate bit string as $\pend = [p_1,...,p_k]$. 
A re-collision occurs at around $r+m$ if and only if $a_1$ is also located at $\pend$ in this round. The probability of this event can be bounded by
the maximum probability of $a_1$ being at any position $\pend = [p_1,...,p_k]$ in round $r+m$ after starting from position $[0,0,...,0]$. We bound this probability below, giving the lemma.
\end{proof}

\begin{lemma}
\label{collideprobboundhc1} 
Consider an agent $a_1$ randomly walking on a $k$-dimensional hypercube with $A=2^k$ vertices, each labeled by a bit string in $\{0,1\}^k$. Assume that $a_1$ starts at position $[0,0,...,0]$. For any $m \ge 0$ and position $\pend = [p_1,...,p_k]$, let $\mathcal{C}$ be the event that $a_1$ is at position $\pend$ after $m$ steps. We have:
\begin{align*}
\Pr[\mathcal{C}] \le (9/10)^{m-1} + \frac{1}{\sqrt{A}}.
\end{align*}
\end{lemma}
\begin{proof}
We will upper  bound the probability that $a_1$ ends at any position $\pend$, denoted  by the bit string $\pend = [p_1,...,p_k]$ after $m$ steps, assuming that $m$ is even. 
If $m$ is odd, the probability of ending at $\pend$ after $m$ steps is equal to the sum over all $\pend'$ with hamming distance $1$ from $\pend$ of the probability  of ending at $\pend'$ after $m-1$ steps. There are $k$ such positions. Further, the probability of moving from any $\pend'$ to $\pend$ in a single random-walk step is $1/k$. Thus, the probability of ending at $\pend$ after $m$ steps is the mean of the probabilities of ending  at each $\pend'$ after $m-1$ steps. This will be bounded by our bound on ending at any  position after $m-1$  steps (noting that $m-1$ is even).

We first argue that 
the number  of paths  that  $a_1$ can take ending in position $\pend = [p_1,...,p_k]$ in round $r+m$ is upper bounded by the number of paths it can take ending at $[0,0,...,0]$. We do this by  describing an injection from paths ending at $\pend$ to paths ending at $[0,0,...,0]$.

Since  $m$ is even,  the parity of $\pend = [p_1,...,p_k]$ is even. We can thus arbitrarily  pair up nonzero indices in this string. For any index pair $(i,j)$, any path ending at $\pend$ must take an odd number of steps in both the $i$ and $j$ directions. By flipping the last step in the path which is in either the $i$ or $j$ direction to the other direction, we obtain a path which makes  an even number of steps in each direction. Further, this path is unique -- given it, we can recover the path ending at $\pend$ simply by again flipping the last step in either direction $i$ or $j$ to be in the other direction. In this way, but iterating through our arbitrary pairings, we can map any path ending at $[p_1,...,p_k]$ to a unique path ending at $[0,...,0]$, giving us the injection. 

We now bound the number of paths ending at $[0,...,0]$, which is identical to
%
  the number of ways $m$ flips can be placed into $k$ buckets, where each bucket has an even number of elements. This quantity is:$$\sum_{\substack{s_1+\ldots +s_k = m \\ (s_i \mod 2 )\equiv 0 }} \frac{m!}{s_1!\cdot \ldots \cdot s_k!}.$$
  
This value is equal to the coefficient of $x^{m}$ in the exponential generating function $$m!\left(1+\frac{x^2}{2!}+\frac{x^4}{4!}+ \ldots\right)^{k} = m!\left(\frac{e^{x} + e^{-x}}{2} \right)^k = \frac{m!}{2^k}\sum_{i=0}^{k} \binom{k}{i} e^{x(2i - k)}.$$
By differentiating $m$ times, we find that the coefficient of $x^{m}$ is: $$\frac{1}{2^k}\sum_{i=0}^{k} \binom{k}{i} (2i-k)^{m} = \sum_{i=0}^{k} \left(\binom{k}{i}/2^{k}\right)\cdot (2i-k)^{m}.$$

This summation is exactly $\E\left [X^{m}\right  ]$, where $X$ is a sum of $k$ i.i.d. random variables each equal to $1$ with probability $1/2$ and $-1$ otherwise. For any $c \in (0,1]$, we can split the expectation:
\begin{align*}
\E\left [X^{m}\right] &= \E \left  [X^{m} | |X| \ge ck\right]\cdot \Pr[|X| \ge ck] + \E \left [X^{m} | |X| \le ck\right ]\cdot \Pr[|X| \le ck] &\\
&\le  k^{m}\cdot  \Pr[|X| \ge ck] + (ck)^{m}.
\end{align*}

To bound the return probability, we divide this count by the the total number of possible paths taken by $a_1$ $m$ steps, $k^{m}$, giving an upper bound of:
\begin{align*}
\Pr[|X| \ge ck] + c^{m}.
\end{align*}

By a Hoeffding bound, $\Pr[|X| \ge ck] \le 2e^{-c^2k/2}$. If we set $c = \sqrt{\ln A/k} = \sqrt{\ln 2}$ then $\Pr[|X| \ge ck] \le 1/\sqrt{A}$. So our final probability bound is:
\begin{align*}
\Pr[|X| \ge ck] + c^{m} \le \frac{1}{\sqrt{A}} + (\sqrt{\ln 2})^{m} < \frac{1}{\sqrt{A}}+(9/10)^m,
\end{align*}
yielding the lemma. Note that in the lemma the stated bound is $\frac{1}{\sqrt{A}}+(9/10)^{m-1}$ to account for odd $m$, as discussed. Also note that, by adjusting $c$, it is possible to trade off the terms in the above bound, giving stronger inverse dependence on $A$ at the expense of slower exponential decay in $m$.
\end{proof}

\paragraph{Density Estimation Bound:}
Plugging Lemma \ref{collideprobboundhc} into Lemma \ref{generalBound}, we have $B(t) = \sum_{m=0}^{t} \beta(m) \leq 10 + t/\sqrt{A}$. If we assume $t = O(\sqrt{A})$, this gives  $\epsilon = O \left ( \sqrt{\frac{\log(1/\delta)}{td}}\right )$ and so
$t = O \left (\frac{\log(1/\delta)}{\epsilon^2d} \right )$, matching independent sampling.

\section{Applications}\label{sec:applications}

We next discuss algorithmic applications of our ant-inspired density estimation algorithm (Algorithm \ref{random_walk_sampling}) and the analysis techniques we develop.

\subsection{Social Network Size Estimation}\label{sec:size}

Random-walk-based density estimation is closely related to work on estimating the size of social networks and other massive graphs using random walks \cite{katzir2014estimating,kurant2012graph,lu2012sampling,lu2014variance}. In these applications, one does not have access to the full graph (so cannot exactly count the nodes), but can simulate random walks by following links between nodes \cite{mislove2007measurement, gjoka2009walk}. One approach is to run a single random walk and count repeat node visits \cite{lu2012sampling,kurant2012graph}. Alternatively, \cite{katzir2014estimating} proposes running multiple random walks and counting their collisions, which gives an estimate of the walk's density. Since the number of walks is known, this yields an estimate for  network size. 

This approach can be significantly more efficient since the dominant cost is typically in link queries to the network. With multiple, shorter random walks, this cost can be trivially distributed to multiple servers simulating walks independently. Visit information can then be aggregated and the collision count can be computed in a centralized manner.

Walks are first run for a `burn-in period' so that their locations are distributed approximately by the stable distribution of the network. The walks are then halted, and the number of collisions in this final round are counted. The collision count gives an estimate of the walks' density. Since the number of walks is known, this yields an estimate for  network size.

We show that ant-inspired algorithms can give runtime improvements over this method. After burn-in, instead of halting the walks immediately, we run them for \emph{multiple rounds}, recording encounter rates as in Algorithm \ref{random_walk_sampling}. This allows the use of fewer walks, decreasing total burn-in cost, and giving faster runtimes when mixing time is relatively slow, as is common in social network graphs \cite{mohaisen2010measuring}.

\subsubsection{Random-Walk-Based Algorithm for Network Size Estimation}

Consider an undirected, connected, non-bipartite graph $G = (V,E)$. Let $S$ be the set of vertices of $G$ that are `known'. Initially, $S = \{ v \}$ where $v$ is a seed vertex. We can access $G$ by looking up the neighborhood $\Gamma(v_i)$ of any vertex $v_i \in S$ and adding $\Gamma(v_i)$ to $S$.

To compute the network size $|V|$, we could scan $S$, looking up the neighbors of each vertex and adding them to the set. Repeating this process until no new nodes are added ensures that $S = V$ and we know the network size. However, this method requires $|V|$ neighborhood queries. The goal is to use significantly fewer queries using random-walk-based sampling.

A number of challenges are introduced by this application. While we can simulate many random walks on $G$, we can no longer assume these random walks start at randomly chosen nodes, as we do not have the ability to uniformly sample nodes from the network. Instead, we must allow the random walks to run for a burn-in phase of length proportional to the mixing time of $G$. After this phase, the walks are distributed approximately according to the stable distribution of $G$.

Further, in general $G$ is not regular. In the stable distribution, a random walk is located at a vertex with probability proportional to its degree. Hence, collisions tend to occur more at higher degree vertices. To correct for this bias, we count a collision at vertex $v_i$ with weight $1/\deg(v_i)$. 

Our results depend on a natural generalization of re-collision probability. For any $i,j$, let $p(v_i,v_j,m)$ be the probability that an $m$-step random walk starting at $v_i$ ends at $v_j$. Define:
\begin{align*}
\beta(m) \eqdef \frac{\max_{i,j} p(v_i,v_j,m)}{\deg(v_j)}.
\end{align*}
Intuitively, $\beta(m)$ is the maximum $m$-step collision probability, weighted by degree since higher degree vertices are visited more in the stable distribution. Let $B(t) = \sum_{m=1}^t \beta(m)$. Note that this weighted $B(t)$ is trivially upper bounded by the unweighted measure used in Lemma \ref{generalBound}.

For simplicity, we initially ignore burn-in and assume that our walks start distributed exactly by the stable distribution of $G$. A walk starts at vertex $v_i$ with probability $p_i \eqdef \frac{\deg(v_i)}{\sum_i \deg(v_i)} = \frac{\deg(v_i)}{2|E|}$ and initial locations are independent. We also assume knowledge of the average degree $\ol \deg = 2|E|/|V|$. 
We later rigorously analyze burn-in and show to estimate $\ol \deg$, completing our analysis.

\begin{algorithm}[H]
\caption{Random-Walk-Based Network Size Estimation}
{\bf input}: step count $t$, average degree $\ol \deg$, $n$ random starting locations $[w_1,...,w_n]$ distributed independently according to the network's stable distribution 
\begin{algorithmic}
\State $[c_1,...,c_n] := [0,0,...,0]$
\For{$r = 1,..., t$}
		\State{$\forall j$, set $w_j := randomElement(\Gamma(w_j))$}\Comment{\textcolor{blue}{$\Gamma(w_j)$ denotes the neighborhood of $w_j$.}}
		\State{$\forall j$, set $c_j := c_j + \frac{count(w_j)}{\deg(w_j)}$}\Comment{\textcolor{blue}{$count(w_j)$ returns the number of other walkers currently at $w_j$.}}
\EndFor
\State{$C := \frac{\ol{\deg}\sum{c_j}}{n(n-1)t}$}\\
\Return{$\tilde A = 1/C$}
\end{algorithmic}
\label{size_estimation_algo}
\end{algorithm}
Note that there are many ways to implement the $count(\cdot)$ function used in Algorithm \ref{size_estimation_algo}. One possibility is to simulate the random walks in parallel, recording their paths, and then to perform centralized post-processing to count collisions. As queries to the network are considered to dominate time cost, this collision counting step is relatively inexpensive.

\begin{theorem}\label{sizeEstimationTheorem} If Algorithm \ref{size_estimation_algo} is run using $n$ random walks for $t$ steps, as long as $n^2 t = \Theta \left (\frac{B(t)\ol \deg + 1}{\epsilon^2 \delta}\cdot |V| \right )$, then with probability at least $1-\delta$, it returns $\tilde A \in \left [(1-\epsilon)|V|, (1+\epsilon)|V|\right ]$.
\end{theorem}

\subsubsection{Analysis of Idealized Algorithm}\label{sec:idealized}
We start with the analysis of Algorithm \ref{size_estimation_algo}, which is given the average degree $\ol{\deg}$ as input and random walk starting locations distributed according to the network's stable distribution.

Throughout this section, we work directly with the weighted total collision count $C = \frac{\ol{\deg}\sum{c_j}}{ n(n-1)t}$, showing that it is close to its expectation with high probability and hence giving the accuracy bound for $\tilde A$. As in the density estimation case, we start by showing that $C$ is correct in expectation. 
\begin{lemma}\label{sizeExpectation} $\E[C] = 1/|V|.$
\end{lemma}
\begin{proof}
Let $c_{j}(r)$ be the number of collisions, weighted by inverse vertex degree, walk $j$ expects to be involved in at round $r$. In each round all walks are at vertex $v_i$ with probability $p_i = \frac{\deg(v_i)}{2|E|}$, so:
\begin{align*}
\E [c_{j}(r)] = \sum_{i=1}^{|V|} \left [ \frac{\deg(v_i)}{2|E|} \cdot \frac{(n-1)\deg(v_i)}{2|E|} \cdot \frac{1}{\deg(v_i)} \right ]= \frac{n-1}{4|E|^2}\sum_{i=1}^{|V|} \deg(v_i) = \frac{n-1}{2|E|}.
\end{align*}
By linearity of expectation, $\E [c_j] = \frac{t(n-1)}{2|E|}$, $\E \left [\sum c_j \right ]= \frac{tn(n-1)}{2|E|}$ and hence, $\E [C] = \frac{\ol \deg}{2|E|} = 1/|V|.$
\end{proof}

We now show concentration of $C$ about its expectation.  Let $c_{i,j}$ be the weighted collision count between walks $w_i$ and $w_j$ where $i \neq j$.
It is possible to follow the moment bound proof of Lemma \ref{per_agent_moments} and bound all moments of $c_{i,j}$. However, there is no clear event we can condition on to ensure independence of all $c_{i,j}$'s. Hence, we cannot prove an analogous bound to Corollary \ref{subExpCorrollary} and employ the concentration result of Lemma \ref{tail_bound}.
Instead, we bound just the second moment (the variance) of each $c_{i,j}$ and obtain our concentration results via Chebyshev's inequality as in Theorem \ref{thm:ringCheby}. This leads to a linear rather than logarithmic dependence on the failure probability $1/\delta$. However, we note that we can simply perform $\log(1/\delta)$ estimates each with failure probability $1/3$ and return the median, which will be correct with probability $1 - \delta$. 

\begin{lemma}[Degree Weighted Collision Variance Bound]\label{per_agent_moments_size_estimation}
For all $i,j \in [1,...,n]$ with $i \neq j$, let $\bar c_{i,j} \eqdef c_{i,j} - \E [c_{i,j}]$. $\E \left [ \bar c_{i,j}^2 \right]  = O\left (\frac{t (B(t) + |V|/|E|)}{|E|}\right )$.
\end{lemma}
\begin{proof}
We can write $\E \left [\bar c_{i,j}^2\right ] = \E \left [c_{i,j}^2\right] -\left (\E \left[c_{i,j}\right]\right)^2 \le \E \left [c_{i,j}\right ]^2 $. We can then split $c_{i,j}$ over rounds to give:
\begin{align*}
\E \left [\bar c_{i,j}^2\right] \le \E &\left [ \left (\sum_{r=1}^t c_{i,j}(r) \right )^2\right ] = \sum_{r=1}^t \E \left [ c_{i,j}(r)^2\right ] + 2 \sum_{r=1}^{t-1} \sum_{r'=r+1}^t \E \left [ c_{i,j}(r) c_{i,j}(r')\right ].
\end{align*}
Since the walks are in the stable distribution, and hence located at $v_i$ in each round with probability $\frac{\deg(v_i)}{2|E|}$, we have the weighted collision $c_{i,j}(r) = \frac{1}{\deg(v_i)}$ with probability  $\frac{\deg(v_i)^2}{(2|E|)^2}$. We thus have $\E \left [ c_{i,j}(r)^2\right ] = \sum_{i=1}^{|V|} \left (\frac{\deg(v_i)^2}{(2|E|)^2} \cdot \frac{1}{\deg(v_i)^2}\right )$. $\E \left [ c_{i,j}(r) c_{i,j}(r')\right ]$ can be computed similarly by summing over all pairs of vertices $\frac{1}{\deg(v) \deg(u)}$ times the probability that the agents collide at vertex $v$ in round $r$ and then again at vertex $u$ in round $r'$. Overall this gives:
\begin{align*}
\E \left  [\bar c_{i,j}^2\right] &\le t \sum_{i=1}^{|V|} \left (\frac{\deg(v_i)^2}{(2|E|)^2} \cdot \frac{1}{\deg(v_i)^2}\right )\\&\hspace{4em}+ 2 \sum_{r=1}^{t-1} \sum_{r'=r+1}^t \left ( \sum_{i=1}^{|V|} \left ( \frac{\deg(v_i)^2}{(2|E|)^2 } \cdot \frac{1}{\deg(v_i)} \cdot \sum_{j=1}^{|V|} \frac{p(v_i,v_j,r-r')^2}{\deg(v_j)}\right ) \right )\\
&\le \frac{t|V|}{4|E|^2} + 2t \sum_{m=1}^{t-1} \left ( \sum_{i=1}^{|V|} \left ( \frac{\deg(v_i)}{(2|E|)^2 } \cdot \beta(m) \sum_{j=1}^{|V|} p(v_i,v_j,m)\right ) \right )
\end{align*}
where in the last step we write $r-r' = m$ and use the fact that $\beta(m) \eqdef \frac{\max_{i,j} p(v_i,v_j,m)}{\deg(v_j)}$. We have $\sum_{j=1}^{|V|} p(v_i,v_j,m) = 1$ and so can simplify the above as:
\begin{align*}
\E \left [\bar c_{i,j}^2\right ] &\le \frac{t|V|}{4|E|^2} + 2t \sum_{m=1}^{t-1} \frac{\sum_{i=1}^{|V|} \deg(v_i)}{(2|E|)^2 }\cdot \beta(m)\\
&= \frac{t|V|}{4|E|^2} + 2t \sum_{m=1}^{t-1} \frac{\beta(m)}{2|E|} = O\left (\frac{t (B(t) + |V|/|E|)}{|E|}\right ).
\end{align*}
\end{proof}

\begin{lemma}[Total Collision Variance Bound]\label{collisionVariance} Let $\ol C = \frac{\ol{\deg}\sum_j \bar c_j}{n(n-1) t}$. $$\E \left [\bar C^2 \right ] = O \left (\frac{1}{n^2t} \cdot \frac{B(t)|E| + |V|}{|V|^2} \right ).$$
\end{lemma}
\begin{proof}
$\sum_{j=1}^n \bar c_j =  \sum_{i, j \in [1,...,n], i\neq j} \bar c_{i,j}.$ We closely follow the variance calculation in \cite{katzir2014estimating}:
\begin{align*}
\E\left [ \left ( \sum_{i,j \in [1,...,n], i\neq j} \bar c_{i,j}\right )^2\right ]  &= \sum_{i,j \in [1,...,n], i\neq j} \left [\sum_{i',j' \in [1,...,n], i\neq j} \bar c_{i,j}\cdot \bar c_{i',j'}\right ]\\
&= 2{n \choose 2} \E \left [ \bar c_{i,j}^2\right  ]+ 4!{n \choose 4}  \E [\bar c_{i,j}]^2 + 2\cdot 3!{n \choose 3} \E [\bar c_{i,j} \bar c_{i,k}].
\end{align*}
The first term corresponds to the cases when $i=i'$ and $j=j'$. The second corresponds to $i \neq i'$ and $j \neq j'$, in which case $\bar c_{i,j}$ and $\bar c_{i',j'}$ are independent and identically distributed. The $4!{n \choose 4}$ multiplier is the number of ways to choose an ordered set of four distinct indices. The last term corresponds to all cases when either $i=i'$ or $j =j'$. There are $3!{n \choose 3}$ ways to choose an ordered set of three distinct indices, multiplied by two to account for the repeated index being in either the first or second position. Using $\E [\bar c_{i,j]}= 0$ and the bound on $\E \left [ \bar c_{i,j}^2  \right]$ from Lemma \ref{per_agent_moments_size_estimation}:
\begin{align}\label{independentZeros}
\E \left [ \left (\sum_{i,j\in [1,...,n], i\neq j} \bar c_{i,j} \right )^2\right ] &= O\left (\frac{n^2 t (B(t)+|V|/|E|)}{|E|}  \right ) + 0 + 2\cdot 3!{n \choose 3} \E [\bar c_{i,j} \bar c_{i,k}].
\end{align}
When $j\neq k$, $\bar c_{i,j}$ and $\bar c_{i,k}$ are independent and identically distributed conditioned on the path that walk $w_i$ traverses (this is similar to the independence used to prove Corollary \ref{subExpCorrollary}). Let $\Psi_i$ be the $t$-step path chosen by $w_i$.
\begin{align}
\E \left [\bar c_{i,j} \bar c_{i,k} \right ] &= \sum_{\psi_i} \Pr \left [\Psi_i = \psi_i \right]\cdot \E \left [  \bar c_{i,j} \middle | \Psi_i = \psi_i \right ] \cdot \E \left [  \bar c_{i,k} \middle | \Psi_i = \psi_i \right ]\nonumber\\
&= \sum_{\psi_i} \Pr \left [\Psi_i = \psi_i \right]\cdot \E \left [  \bar c_{i,j} \middle | \Psi_i = \psi_i \right ]^2\nonumber\\
&= \sum_{\psi_i} \Pr \left [\Psi_i = \psi_i \right]\cdot \left (\E \left [ c_{i,j} \middle | \Psi_i = \psi_i \right ] - \E \left [c_{i,j} \right ]\right )^2.\label{expect0}
\end{align}
$\E \left [ c_{i,j} \middle | \Psi_i = \psi_i \right ] = \sum_{r = 1}^t \frac{\deg(\psi_i(r))}{2|E|} \cdot \frac{1}{\deg(\psi_i(r))} = \frac{t}{2|E|} = \E\left [c_{i,j} \right ]$. That is, the expected number of collisions is identical for every path of $w_i$. Plugging into \eqref{expect0}, $$\E \left [\bar c_{i,j} \bar c_{i,k} \right ] = 0.$$
So finally, plugging back into equation \eqref{independentZeros}, $$\E\left [ \left (\sum_{i,j\in [1,...,n], i\neq j} \bar c_{i,j} \right )^2\right ] = O\left (\frac{n^2 t (B(t) + |V|/|E|)}{|E|}  \right ) $$ and thus:
\begin{align*}
\E\left [\ol{C}^2\right] &= O\left (\frac{n^2 t(B(t) +|V|/|E|)}{|E|}  \cdot \left(\frac{\ol{\deg}}{ n(n-1)t}\right )^2\right )\\
 &= O \left (\frac{1}{n^2t} \cdot \frac{(B(t)+|V|/|E|) \cdot |E|}{|V|^2} \right )\\
&= O \left (\frac{1}{n^2t} \cdot \frac{B(t)|E| + |V|}{|V|^2} \right ).
\end{align*}
\end{proof}
With this variance bound in place, we can finally prove Theorem \ref{sizeEstimationTheorem}. 
\begin{proof}[Proof of Theorem \ref{sizeEstimationTheorem}]
Note that $\bar C = C - \E [C]$ and by Lemma \ref{sizeExpectation}, $\E [C] = 1/|V|$. By Chebyshev's inequality Lemma \ref{collisionVariance} gives:
\begin{align*}
\Pr \left [\left |C -  \E [C]\right | \ge \epsilon \E [C]\right ] \le \frac{1}{\epsilon^2 n^2 t} \cdot (B(t)|E| + |V|).
\end{align*}
Rearranging gives us that, in order to have $C \in \left [\frac{1-\epsilon}{|V|}, \frac{1+\epsilon}{|V|}\right]$ with probability $\delta$, we must have:
\begin{align*}
n^2 t = \Theta \left (\frac{B(t)|E|  + |V|}{\epsilon^2 \delta}\right ).
\end{align*}
Since $\tilde A = 1/C$, if $C \in  \left [\frac{1-\epsilon}{|V|}, \frac{1+\epsilon}{|V|}\right]$ then $\tilde A \in \left [\frac{|V|}{1+\epsilon}, \frac{|V|}{1-\epsilon} \right ] \subseteq \left [(1-2\epsilon) |V|, (1+2\epsilon) |V| \right ]$ as long as $\epsilon < 1/2$. This gives the theorem after adjusting constants on $\epsilon$ and recalling that $\ol{\deg} = |E|/|V|$.
\end{proof}

\subsubsection{Estimating The Average Degree}

We now show how to estimate the value of $\ol{\deg}$ used in Algorithm \ref{size_estimation_algo}. Specifically, we need a $(1\pm\epsilon)$ approximation to $\frac{1}{\ol{\deg}}$. If we then substitute this into the formula $\tilde A = \frac{\sum_j c_j}{\ol{\deg}\cdot n(n-1)t}$, we still have a $(1\pm O(\epsilon))$ approximation to the true network size. We use the algorithm and analysis of \cite{katzir2014estimating}, which gives a simple approximation via inverse degree sampling.
\begin{algorithm}[H]
\caption{Average Degree Estimation}
{\bf input}: $n$ random starting locations $[w_1,...,w_n]$ distributed independently according to the network's stable distribution.
\begin{algorithmic}
\State{$\forall j$, set $d_j := \frac{1}{\deg(w_j)}$}\Comment{\textcolor{blue}{Sampling}}\\
\Return{$D := \frac{\sum{d_j}}{n}$}
\end{algorithmic}
\label{degree_estimation_algo}
\end{algorithm}

\begin{theorem}[Average Degree Estimation]\label{avgDegreeEst} If $n = \Theta \left (\frac{1}{\epsilon^2 \delta}\cdot \frac{\ol \deg}{\deg_{\min}} \right )$, Algorithm \ref{degree_estimation_algo} returns $D$ such that, with probability at least $1-\delta$, $D \in \left [ \frac{1-\epsilon}{\ol \deg}, \frac{1+\epsilon}{\ol\deg} \right ]$.
\end{theorem}
\begin{proof}
Using that in the stable distribution a walk is at vertex $v_i$ with probability $\frac{\deg(v_i)}{2|E|}$ we have:
\begin{align*}
\E [D] = \frac{1}{n} \sum_{j=1}^n \E [d_j] =  \frac{1}{n} \cdot n \cdot \sum_{i=1}^{|V|} \left (\frac{\deg(v_i)}{2|E|} \cdot \frac{1}{\deg(v_i)}\right ) = \frac{|V|}{2|E|} = \frac{1}{\ol \deg}.
\end{align*}
For each $d_j$ let $\bar{d}_j = d_j \E [d_j]$. We have $\E[\bar{d}_j^2] = \E[{d}_j^2] - \E[{d}_j]^2 \le \E[{d}_j^2] $. We can explicitly compute this expectation as:
\begin{align*}
 \E[{d}_j^2]  = \sum_{i=1}^{|V|} \frac{\deg(v_i)}{2|E|} \frac{1}{\deg(v_i)^2} \le \frac{|V|}{2|E|\deg_{\min}} = \frac{1}{\deg_{\min}} \cdot \frac{1}{\ol \deg}.
\end{align*}

Additionally, since each $d_j$ is independent and identically distributed, and since $\bar D = \frac{1}{n}\sum d_j$, letting $\bar D = D - \E [D]$,
\begin{align*}
\E\left [ \bar D^2\right ] = \frac{1}{n} \E [\bar d_j^2] &\le \frac{1}{n} \E [d_j^2]
= \frac{1}{n\deg_{\min}} \cdot \frac{1}{\ol \deg}
\end{align*}
Applying Chebyshev's inequality and the fact that $\E [D] = \frac{1}{\ol\deg}$:
$$
\Pr \left [ \left |D - \E [D] \right | \le \frac{\epsilon}{\ol \deg}\right ] \le \frac{\ol\deg}{\epsilon^2 n\deg_{min}}
.$$ Rearranging, to succeed with probability $\ge 1-\delta$ it suffices to set
$
n = \Theta \left (\frac{1}{\epsilon^2 \delta}\cdot \frac{\ol \deg}{\deg_{\min}} \right ).
$
\end{proof}
\subsubsection{Handling Burn-In Error}\label{sec:burnin}

Finally, we remove our assumption that walks start distributed exactly according to the network's stable distribution, rigorously bounding the length of burn-in required before running Algorithm \ref{size_estimation_algo}. 

Let $\mathcal{D}^* \in \mathbb{R}^{|V|^n}$ be a vector representing the true stable distribution of $n$ random walks on $G$ and $\mathcal{D}_t\in \mathbb{R}^{|V|^n}$ be a vector representing the distribution of the walks after running for $t$ burn-in steps. Specifically, each walk $w_1,...,w_n$ is initialized at a single seed vertex $v$. For $t$ rounds we then update the location of each walk independently by moving to a randomly chosen neighbor. 
Both vectors are probability distributions: they have all entries in $[0,1]$ and $\norm{\mathcal{D}^*}_1 = \norm{\mathcal{D}}_1 = 1$.

Let $\Delta = \mathcal{D}^*-\mathcal{D}_t$ and assume that $\norm{\Delta}_1 \le \delta$.
We can consider two equivalent algorithms: draw an initial set of locations $W = w_1,...,w_n$ from $\mathcal{D}^*$, run Algorithm \ref{size_estimation_algo}, and then artificially fail with probability $\max\{0,\Delta(W)\}$. Alternatively, draw $W = w_1,...,w_n$ from $\mathcal{D}_t$, run Algorithm \ref{size_estimation_algo}, and then artificially fail with probability $\max\{0,-\Delta(W)\}$. These algorithms are clearly equivalent. The first obtains a good estimator with probability $1-2\delta$: probability $\delta$ that Algorithm \ref{size_estimation_algo} fails when initialized via the stable distribution $\mathcal{D}^*$ by Theorem \ref{sizeEstimationTheorem} plus an artificial failure probability of $\le \norm{\Delta}_1 \le \delta$. The second then clearly also fails with probability $2\delta$. This can only be higher than if we did not perform the artificial failure after running Algorithm \ref{size_estimation_algo}. Therefore, running Algorithm \ref{size_estimation_algo} with a set of random walks initially distributed according to $\mathcal{D}_t$ yields success probability $ \ge 1-2\delta$.

How long must the burn-in period be to ensure $\norm{\mathcal{D}^*-\mathcal{D}_t}_1 \le \delta$? Let $\bv{W}$ be the random walk matrix of $G$. Let $\lambda_1 \geq \lambda_2 \geq \ldots \geq \lambda_A$ be the eigenvalues of $\bv{W}$ and $\lambda = \max\{|\lambda_2|, |\lambda_{|V|}| \}$. Let $\mathcal{C}_t \in \mathbb{R}^{|V|}$ denote the location distribution for a single random walk after burn-in and $\mathcal{C}^* \in \mathbb{R}^{|V|}$ denote the stable distribution of a single random walk. 
If we have, for all $i$, $|\mathcal{C}_t(v_i) - \mathcal{C}^*(v_i)| \le \delta/n\cdot \mathcal{C}^*(v_i)$ then for any $W$:
\begin{align*}
|\mathcal{D}_t(W) - \mathcal{D}^*(W)| &= \left| \prod_{i=1}^n \mathcal{C}_t(w_i) - \prod_{i=1}^n \mathcal{C}^*(w_i)\right|\\
&\le \prod_{i=1}^n (\mathcal{C}^*(w_i) + \delta/n\cdot \mathcal{C}^*(w_i)) - \prod_{i=1}^n \mathcal{C}^*(w_i)\\
&<  \mathcal{D}^*(W) \sum_{i=1}^n {n \choose i} (\delta/n)^i \le \mathcal{D}^*(W)\sum_{i=1}^n\delta^i \le 2\delta\cdot \mathcal{D}^*(W),
\end{align*}
as long as $\delta < 1/2$. This multiplicative bound gives $\norm{\mathcal{D}^*-\mathcal{D}_t}_1 \le 2\delta$. By standard mixing time bounds (\cite{lovasz1993random}, Theorem 5.1), $|\mathcal{C}_t(v_i) - \mathcal{C}^*(v_i)| \le \frac{\delta}{n|E|} \cdot \mathcal{C}^*(v_i)$ for all $i$ after $M = O\left (\frac{\log(n|E|/\delta)}{1-\lambda}\right) = O\left (\frac{\log(|E|/\delta)}{1-\lambda}\right)$ burn-in steps (since $n < |E|$ or else we could have scanned the full graph.)

\subsubsection{Overall Runtime and Comparison to Previous Work}

Let $M = O\left(\frac{\log(|E|/\delta)}{1-\lambda} \right)$ denote the burn-in time required before running Algorithm \ref{size_estimation_algo}. In order to obtain a $(1\pm \epsilon)$ estimate of network size with probability $1-\delta$ we must run $n$ random walks for $M+t$ steps, making $n(M+t)$ link queries, where by Theorems \ref{sizeEstimationTheorem} and \ref{avgDegreeEst}:
\begin{align}\label{tBound2}
n = \Theta \left (\max \left \{ \frac{\ol \deg}{\deg_{\min}\epsilon^2\delta},\sqrt{\frac{|V|\cdot (B(t)\ol \deg+1)}{t \cdot \epsilon^2 \delta}}\right \} \right ).
\end{align}
Typically, the second term dominates since $\ol \deg << |V|$. Hence, by increasing $t$, we are able to use fewer random walks, significantly decreasing the number of link queries if $M$ is large.

\cite{katzir2014estimating} uses a different approach, halting random walks and counting collisions immediately after burn-in. For reasonable node degrees they require
$
n = \Theta \left (\frac{|V| \cdot \ol \deg}{\epsilon^2 \delta \cdot \sqrt{\sum \deg(v_i)^2} }\right ).
$
Assuming that $\sqrt{\sum \deg(v_i)^2} < n$, and setting $t = 1$, this is somewhat smaller than our bound as $\sum \deg(v_i)^2 \ge |V| \cdot \ol \deg$.
However, \eqref{tBound2} gives an important tradeoff -- by increasing $t$ we can increase the number of steps in our random walks, decreasing the total number of walks.

As an illustrative example, consider a $k$-dimensional torus graph for $k \ge 3$ (for $k=2$ mixing time is $\Theta(|V|)$ so we might as well census the full graph). 
The burn-in mixing time required for Algorithm \ref{size_estimation_algo} is $M = \Theta (\log(|V|/\delta) |V|^{2/k})$. All nodes have degree $2k$, and using the bounds above, to obtain a $(1\pm \epsilon)$ estimate of $|V|$, the algorithm of \cite{katzir2014estimating} requires 
$$M\cdot n = \Theta \left (\frac{\log(|V|/\delta)}{\epsilon \sqrt{d}} \cdot |V|^{2/k + 1/2} \right )$$ 
link queries.
In contrast, assuming $|V|$ is large, we require
$
n = \Theta \left (\sqrt{\frac{|V|}{t \cdot \epsilon^2 \delta}} \right )
$
since by Lemma \ref{collideprobboundd}, $B(t) = O(1/k)$ and $\ol \deg = \deg_{\min} = k$. If we set $t = \Theta(M)$, the total number of link queries needed is
$$
n(M+t) = O \left (\frac{\sqrt{\log(|V|/\delta)}}{\epsilon \sqrt{d}} \cdot |V|^{(k+1)/2k} \right ).
$$
This beats \cite{katzir2014estimating} by improving dependence on $|V|$ and the logarithmic burn-in term. Ignoring error dependencies, if $k = 3$, \cite{katzir2014estimating} requires $\Theta(n^{7/6})$ queries which is more expensive than fully censusing the graph. We require $O(n^{2/3})$ queries, which is sublinear in the graph size.
 
We leave open comparing our bounds with those of \cite{katzir2014estimating} on more natural classes of graphs. It would be interesting to determine typical values of $B(t)$ in real work networks or popular graph models, such as preferential attachment models and others with power-law degree distributions.

\subsection{Distributed Density Estimation by Robot Swarms}

Algorithm \ref{random_walk_sampling} can be directly applied as a simple and robust density estimation algorithm for robot swarms moving on a two-dimensional plane modeled as a grid. Additionally, the algorithm can be used to estimate the frequency of certain properties within the swarm. Let $d$ be the overall population density and $d_P$ be the density of agents with some property $P$. Let $f_P = d_P/d$ be the relative frequency of $P$.

Assuming that agents with property $P$ are distributed uniformly in population and that agents can detect this property (through direct communication or some other signal), then they can separately track encounters with these agents. They can compute an estimate $\tilde d$ of $d$ and $\tilde d_P$ of $d_P$. By Theorem \ref{naturalAlgoThm}, after running for $t = \Theta \left (\frac{\log(1/\delta)[\log\log(1/\delta)+\log(1/d\epsilon)]^2}{d_P \epsilon^2} \right )$ steps, with probability $1-2\delta$, 
$$\tilde d_P/\tilde d \in \left [\left (\frac{1-\epsilon}{1+\epsilon}\right ) f_P, \left (\frac{1+\epsilon}{1-\epsilon}\right) f_P \right ] = \left [(1-O(\epsilon))f_P,(1+O(\epsilon))f_P \right ]$$ for small $\epsilon$.

In an ant colony, properties may include whether or not an ant has recently completed a successful foraging trip \cite{gordon1999interaction}, or if an ant is a nestmate or enemy \cite{adams1990boundary}. In a robotics setting, properties may include whether a robot is part of a certain task group, whether it has completed a certain task, or whether it has detected a certain event or environmental property.

\section{Discussion and Future Work}\label{sec:futureWork}

We have presented a theoretical analysis of random-walk-based density estimation by agents on a two-dimensional torus graph. We have also presented applications of our techniques to density estimation on other regular graph topologies and to the problems of social network size estimation and density estimation on robot swarms.
Our work leaves open a number of open questions which we discuss below.

\subsection{Extensions to Our Model}\label{sec:antExtensions}

We feel that our simple model of density estimation on the two-dimensional torus (Section \ref{sec:model}) well reflects the behavior of ants estimating density via collision rates while moving around a two-dimensional surface. However, extending our results to more realistic models would be a very interesting direction. 

We believe that it is important to consider a model in which agents are not positioned uniformly at random on the torus. Without the uniform placement assumption, solving the global density estimation problem that we have defined may become difficult. If most agents are placed in a very small portion of the torus, any other agent initially located far away from these agents must traverse a large  portion of the torus to find them with good probability, and hence, to estimate the global population density accurately.

There are several ways to overcome this difficulty. First, given some distribution of the agents over the torus, it may be possible to give bounds parameterized by the distance from this distribution to the uniform distribution. If the distribution is close to uniform, random-walk-based estimation should do a good job estimating density. If it is very far from uniform (e.g., in the example above, where many  agents are concentrated in a small area), global density  estimation will become more difficult.

Alternatively, as discussed in Section \ref{sec:problemD}, it would be very  interesting to formally define a \emph{local density estimation problem}, which takes an agent's initial location into account when defining the density which they aim to estimate. In such a problem, agents located in more densely populated areas of the torus will return higher local density estimates.

Another interesting direction is to modify our assumption that agents move via random walk, considering more complex models based on empirical studies of ant movement \cite{gordon1993function,nicolis2005effect,boczkowski2017limits}. It may be interesting to study a model in which agents sense and sometimes avoid collisions, or in which they move away from previously encountered ants. It  may also be interesting to consider random-walk-based models, but with asynchronous movement, or 
continuous movement along a surface. Empirically, while there is some work directly testing the assumption of random movement  by considering re-collision rates \cite{boczkowski2017limits}, providing further evidence of how closely our model predicts re-collision probabilities and, in turn, density  estimation accuracy, would be very  interesting.

Finally, it would be valuable to explore the robustness of random-walk-based density  estimation to noise and other perturbations. One possibility is to model noisy collision detection, in which each collision is only detected with some probability, or  in which spurious collisions may occasionally be detected. We may also model noise in ant behavior. For example, each agent may not move via pure random walk, but via some perturbed behavior which assigns nonuniform probabilities to the steps $\{(0,1),(0,-1),(1,0),(-1,0), (0,0)\}$.

\subsection{Biological Applications}

As discussed in Section \ref{sec:intro}, density  estimation is used as a subroutine in many ant behaviors such as quorum sensing \cite{pratt2005quorum} in house-hunting, task allocation \cite{gordon1999interaction,schafer2006forager}, and appraisal of enemy colony strength \cite{adams1990boundary}.
Modeling these behaviors theoretically, and studying how our approximate density estimation results can be composed with higher level algorithms  is a very  interesting direction.

In our own work, we have considered the use of density estimation in the house-hunting process in \emph{Temnothorax} ants 
\cite{radeva2017symbiotic,RadevaML-bda17}, demonstrating that approximate density estimation , where the density estimate  is correct in expectation and within a $(1\pm \epsilon)$ factor of the true density  with  high probability, suffices for efficient decision making in the house-hunting process. It would be interesting, for example, to prove similar results  for task-allocation, where density estimation may be used to approximate the number of workers  currently performing a given task. 

Density estimation behavior may also be used in other species. For example, there is evidence that higher work density stimulates certain reproductive behaviors in honeybee colonies  \cite{smith2017cues}. Studying theoretically  how bees estimate and respond to increased density, and how this behavior compares to ant colony behavior would be  valuable.

Finally, we note that the accuracy bound of Theorem \ref{naturalAlgoThm} depends inversely on the density $d$, and so becomes large  when $d$ is small. In many of the above biological applications, such as in quorum sensing for decision making in ant colonies, agents only need to detect when $d$ is above some fixed threshold. In this case, better bounds, where $t$ depends not on the true density, but just on this detection threshold, may be possible. Additionally, it may  be interesting to understand how multiple agents with different  density estimates can cooperate to learn if a density  threshold has been reached, with more accuracy than if just a single agent  were attempting to detect such a threshold.

\subsection{Algorithmic Applications}

We conclude by discussing algorithmic applications of our analysis techniques, extending the results presented in Section \ref{sec:applications}.

\subsubsection{Random-Walk-Based Sensor Network Sampling}
We believe our moment bounds for a single random walk (Corollaries \ref{visitMomentBound} and \ref{equalizationMomentBound}) can be applied to random-walk-based distributed algorithms for sensor network sampling.
Random-walk-based sensor network sampling \cite{lima2007random,avin2004efficient} is a technique in which
a query message (a `token') is initially sent by a base station to some sensor. The token is relayed randomly between sensors, which are connected via a grid communication network, and its value is updated appropriately at each step to give an answer to the query. 
This scheme is robust and efficient - it easily adapts to node failures and does not require setting up or storing spanning tree communication structures.

Random-walk-based sampling could be used, for example, to estimate the percentage of sensors that have recorded a specific condition, or the average value of some measurement at each sensor. However,
as in density estimation, unless an effort is made to record which sensors have been previously visited, additional error may be added due to repeat visits.
Recording previous visits introduces computational burden -- either the token message size must increase or nodes themselves must remember which tokens they have seen.
We are hopeful that our moment bounds can be used to show that this is unnecessary -- due to strong local mixing, the number of repeat sensor visits will be low, and the performance reduction limited.

We remark that estimating the percentage of sensors in a network or the density of robots in a swarm with a property that is uniformly distributed is a special case of a more general \emph{data aggregation} problem: each agent or sensor holds a value $v_i$ drawn independently from some distribution $\mathcal{D}$. The goal is to estimate some statistic of $\mathcal{D}$, such as its expectation. In the case of density estimation, $v_i$ is simply an indicator random variable which is $1$ with probability $d$ and $0$ otherwise.
Extending our results to more general data aggregation problems and showing that random walk sampling matches independent sampling in some cases is
an interesting future direction.

\subsubsection{Size Estimation of Realistic Networks}

We leave open studying the effectiveness of the algorithm for social network size estimation presented in Section \ref{sec:size} in real world networks, or on popular random graph models for social networks \cite{newman2002random}. It may also be interesting to give bounds for the algorithm in general graphs, parameterized by the global mixing time, rather than the $m$-step recollision probability  $\beta(m)$. While such bounds may give a coarser characterization of the algorithm's performance, they  could be used to compare again worst  case bounds for related random-walk-based network size estimation approaches parameterized by the mixing time \cite{ben2018estimating}.

\subsubsection{Beyond Encounter Rate}

In social network size estimation, robot swarm density  estimation, and sensor network sampling, it is possible to leverage more information than just the random walk encounter rate. For example, a size estimation algorithm can store each agent's full $t$-step path, and count  the number of intersections between these paths. A robot swarm algorithm may assign ids to each agent and use them to identify repeat collisions. It would be valuable to understand if these strategies can be used to improve our bounds, or if they do not give significant advantages. 

\subsubsection{Other Potential Applications}

Finally, there are many potential applications of random-walk-based density  estimation and sampling that we have not yet considered. For example, density estimation may be employed as a subroutine in swarm robot coverage and exploration routines, which aim to explore an unknown environment, or survey a known environment, as quickly  as possible \cite{burgard2000collaborative,batalin2002spreading}. It may be interesting to use density  estimation to detect regions with high robot density, and to then spread out this density to more efficiently distribute exploration. Similar techniques may  be interesting in robot formation problems in which the goal is the spread a swarm of robots or sensors regularly across a surface (or according to some specified distribution)  using a distributed algorithm \cite{glavaski2003vehicle,cortes2004coverage,gilbert2009self}.

\section{Acknowledgements}
We thank Amartya Shankha Biswas, Christopher Musco, and Mira Radeva for useful discussions. We also thank Yehuda Afek, Ziv Bar-Joseph, and Amos Korman for many helpful comments and suggestions. We thank Yury Polyanskiy for pointing our a bug in our original proofs, which assumed independence of collision counts between agents. This bug has been corrected in this writeup. Research was supported by NSF Grants BIO-1455983 and CCF-1461559, NSF CSoI grant CCF-0939370, and AFOSR grant FA9550-13-1-0042. Cameron Musco was partially supported by an NSF graduate student fellowship.

\clearpage

\bibliographystyle{alpha}
\bibliography{densityEstimation}

\appendix

\section{Independent-Sampling-Based Density Estimation}\label{sec:independent}

Here we show that, if agents are not restricted to random walking, but can instead take arbitrary steps in each round, they can avoid collision correlations by splitting into `stationary' and `mobile' groups and counting collisions only between members of different groups. This allows them to essentially simulate independent sampling of grid locations to estimate density. This algorithm is not `natural' in a biological sense, however it is easy to analyze and gives slightly better bounds than the random-walk-based approach (Theorem \ref{naturalAlgoThm}). We give pseudocode in Algorithm \ref{independent_sampling}. Recall that $position$ is an ordered pair denoting an agent's $(x,y)$ coordinates on the torus graph, and $count(position)$ returns the number of other agents at the current position.

\begin{algorithm}[H]
\caption{Independent-Sampling-Based Density Estimation}
Each agent independently executes:
\begin{algorithmic}
\State Set $c := 0$ and with probability $1/2$, $state := walking$, else $state:= stationary$.
\For{$r = 1,..., t$}
\If {$state := walking$}
	\State $position := position + (0,1)$\Comment{\textcolor{blue}{Deterministic walk step.}}
\EndIf
\State{$c := c + count(position) $}\Comment{\textcolor{blue}{Update collision count.\hspace{.5em}}}
\EndFor
\State{$c := c \pmod t$}\\
\Return{$\tilde d = \frac{2c}{t}$}
\end{algorithmic}
\label{independent_sampling}
\end{algorithm}

\subsection{Independent Sampling Accuracy Bound}
Our main accuracy bound for the independent sampling algorithm is given below.

\begin{theorem}[Independent Sampling Accuracy Bound]\label{optimalAlgoThm} After running for $t$ rounds, assuming $t < \sqrt{A}$ and $d \le 1$, an agent executing Algorithm \ref{independent_sampling} returns $\tilde d$ such that, for any $\delta > 0$, with probability $\ge 1-\delta$, $\tilde d \in [(1-\epsilon ) d,  (1+\epsilon) d ]$ for $\epsilon = O\left (\sqrt{\frac{\log(1/\delta)}{td}} \right )$. In other words, for any $\epsilon, \delta \in (0,1)$ if $t = \Theta \left (\frac{\log(1/\delta)}{d\epsilon^2} \right )$, $\tilde d$ is a $(1\pm \epsilon)$ multiplicative estimate of $d$ with probability $\ge 1-\delta$.
\end{theorem}

\begin{proof}

Our analysis is from the perspective of an agent with $state = walking$. By symmetry, the distribution of $\tilde d$ is identical for walking and stationary agents, so considering this case is sufficient.

Initially, assume that no two walking agents start in the same location.
Given this assumption, we know that a walking agent \emph{never collides with another walking agent} -- by assumption they all start in different positions and update these positions identically in each round. In the written implementation, agents always step up, however any fixed pattern (e.g. a spiral) suffices.

Further, assume that agents do not execute the step of setting $c := c \pmod t$. This step will be used to handle walking agents which start at the same location, and will be analyzed at the end of the proof.

In $t$ steps, a walking agent visits $t$ unique squares (here we use the assumption that $t < \sqrt{A}$, the diameter of the grid). Each of the $n$ other agents is located in this set of squares \emph{and} stationary with probability $\frac{t}{2A}$. Further,
each of these events is entirely independent from the rest, as the agents are positioned and choose their state independently. So,  for a walking agent, $c$ is just a sample of $n$ independent random coin flips, each with success probability $\frac{t}{2A}$. Clearly, $\E c = n \cdot \frac{t}{2A} = \frac{td}{2}$ so $\E \tilde d = \E \frac{2c}{t} = d$. Further, by a Chernoff bound, for any $\epsilon \in (0,1)$, the probability that $\tilde d$ is not a $(1 \pm \epsilon)$ multiplicative estimate of $d$ is:
\begin{align*}
\delta = \Pr \left [ |\tilde d - d| \ge \epsilon d \right ] = \Pr \left [ |c - \E c| \ge \epsilon \E c \right ] \le 2e^{-\epsilon^2 \E c / 3} \le 2e^{-\epsilon^2 td/6}.
\end{align*}
This gives:
$\log(1/\delta) \ge \epsilon^2 td/6$ so $\epsilon = O\left ( \sqrt{\frac{\log(1/\delta)}{td}} \right )$, yielding the result.

We now remove the assumption that no two walking agents start in the same location by considering the step where each agent sets $c := c \pmod t$ before returning $\tilde d = \frac{2c}{t}$. If an agent starts alone and is involved in $ < t$ collisions, this operation has no effect -- the above bound holds.

If a walking agent is involved in $< t$ `true collisions' but starts in the same position as $w \ge 1$ other walking agents, the agents move in lockstep throughout the algorithm and are involved in $w \cdot t$ `spurious collisions' ($w$ in each round). Setting $c := c \pmod{t}$ exactly corrects for these spurious collisions and since $c$ now only includes collisions with stationary agents, the bound above holds.

Finally, if an agent is involved in $\ge t$ true collisions, this modification cannot worsen their estimate. If $c \ge t$ and the agent does not set $c:= c \pmod{t}$, they compute $\tilde d \ge \frac{2t}{t} \ge 2$. For $\epsilon < 1$, the agent fails since $d\le 1$. So setting $c:= c \pmod{t}$ can only increase  the probability of success.
\end{proof}

\end{document}